\pdfoutput=1
\documentclass[journal]{IEEEtran}
\usepackage[T1]{fontenc}
\usepackage{cite}
\usepackage{amsmath,amssymb,amsfonts}
\usepackage{algorithm}
\usepackage{algpseudocode}
\usepackage{graphicx}
\usepackage{textcomp}
\usepackage{url}

\newtheorem{theorem}{Theorem}
\newtheorem{example}{Example}
\DeclareMathOperator{\MCP}{MCP}
\DeclareMathOperator{\FPR}{FPR}
\newcommand{\GF}{\mathrm{GF}}
\newcommand{\cL}{\mathcal{L}}
\newcommand{\defeq}{:=}
\DeclareRobustCommand{\thmoneref}{\ref{thm:rand_collinearity}}

\usepackage{nicefrac}

\usepackage{multirow}
\usepackage{booktabs}
\usepackage{makecell}

\usepackage[bookmarks=false,colorlinks=true,linkcolor=blue,citecolor=blue,urlcolor=black]{hyperref}
\hypersetup{pdftitle={Lagrange Interpolation-Based Multi-Bit Watermark for Large Language Models},
  pdfauthor={Jaros\l{}aw Janas, Pawe\l{} Morawiecki, Josef Pieprzyk}}

\begin{document}

\title{Lagrange Interpolation-Based Multi-Bit Watermark for Large Language Models}

\author{Jarosław~Janas\textsuperscript{1}, Paweł~Morawiecki\textsuperscript{1}, and~Josef~Pieprzyk\textsuperscript{1,2}\endgraf
{\small\vspace{0.8ex}%
\textsuperscript{1}Institute of Computer Science, Polish Academy of Sciences, 01-248 Warsaw, Poland\\
\textsuperscript{2}CSIRO, Marsfield, NSW 2122, Australia\endgraf\vspace{3.5ex}%
Corresponding author: Jarosław Janas (jaroslawjanas.dev@gmail.com)\\
This work was supported in part by the National Science Centre, Poland, under Grant OPUS 2023/49/B/ST6/02580.\\
Published in \emph{IEEE Access}, vol.~14, pp.~177747--177765, 2026, doi:~\href{https://doi.org/10.1109/ACCESS.2026.3751996}{10.1109/ACCESS.2026.3751996}.\endgraf}}
\IEEEaftertitletext{\vspace{-1\baselineskip}}

\markboth{J.~Janas \MakeLowercase{\textit{et al.}}: Lagrange Interpolation-Based Multi-Bit Watermark for Large Language Models}%
{J.~Janas \MakeLowercase{\textit{et al.}}: Lagrange Interpolation-Based Multi-Bit Watermark for Large Language Models}

\maketitle

\begin{abstract}
The rapid advancement of LLMs (Large Language Models) has established them as a foundational technology for many AI- and ML-powered human–computer interactions. A critical challenge in this context is the attribution of LLM-generated text --- for example, identifying the specific language model that generated it or the individual user who prompted the model. This capability is essential for combating misinformation, fake news, misinterpretation, and plagiarism. One of the key techniques for addressing this challenge is digital watermarking. This work presents a watermarking scheme for LLM-generated text based on Lagrange interpolation, enabling the recovery of a multi-bit watermark even when the text has been redacted by an adversary. The core idea is to embed a continuous sequence of points $(x, f(x))$ that lie on a single straight line. During extraction, the algorithm recovers the original points along with many spurious ones, forming an instance of the Maximum Collinear Points (MCP) problem, which can be solved efficiently. Experimental results demonstrate that the proposed method is scalable and effective, allowing the embedding of a multi-bit watermark. 
Experiments on Llama 2-7B    
demonstrate that the proposed method achieves match rates of 98.6\%, 96.0\%, 93.9\%, and 85.9\% for watermark sizes of 16, 20, 24, and 32 bits respectively — competitive with state-of-the-art methods. Crucially, watermark extraction completes in a fraction of a second, and the scheme is lightweight and applicable to multiple LLM architectures without model retraining.
\end{abstract}


\section{Introduction}
\label{sec:introduction}
Artificial Intelligence (AI) and Machine Learning (ML) have advanced significantly in recent years and
have given rise to powerful models based on deep learning. In particular, Deep Neural Networks (DNNs) have enabled machines to perform complex cognitive tasks such as image recognition, speech processing, and natural language understanding with accuracy almost equivalent to human performance. A pivotal advancement in this domain is the emergence of Large Language Models (LLMs). They can be seen as deep learning architectures trained on massive textual datasets, which can generate coherent, contextually relevant, and human-like text. LLMs operate by learning rich representations of language that capture its semantic, syntactic, and pragmatic patterns. These models, often based on the Transformer architecture, are pre-trained on diverse corpora and fine-tuned for specific tasks such as question answering, summarisation, translation, and dialogue. Their capabilities stem from their scale, both in terms of parameters and data, which allows them to generalise across a wide range of linguistic contexts.
Recent years have witnessed growing interest in watermarking techniques for generative AI systems as a means of provenance verification and misuse mitigation~\cite{kirchenbauer2023watermark,aaronson2023watermark}.

The growing adoption of LLMs offers a wide array of practical benefits, including automated text correction, assistance with writing, and content refinement. However, alongside these advantages lies a more troubling aspect: the potential misuse of LLMs. One such concern arises when individuals present AI-generated content as their own, falsely claiming authorship. This challenge naturally raises an important question: how can we embed a watermark into LLM-generated text that allows a recipient to reliably trace its origin, whether to a specific user or to the LLM system itself?
The objective of such watermarking is to make it possible to attribute the text to its true source, thereby establishing authorship beyond reasonable doubt. This capability is especially critical in combating misinformation, propaganda, automated plagiarism, and the spread of fake news. 

The remainder of this paper is organised as follows.
Section~\ref{sec:overview} reviews watermarking for LLMs and identifies the gaps we address, and Section~\ref{sec:building_blocks} recalls Lagrange interpolation
over $GF(2^n)$ together with the token-level embedding primitive.
Section~\ref{sec:watermarking_scheme} presents the proposed scheme and its MCP-based extraction, and Section~\ref{sec:security_evaluation} bounds the
probability of spurious collinearity. Section~\ref{sec:experimental_results}
reports the experimental evaluation, Section~\ref{sec:adversarial_attacks} analyses substitution, insertion and deletion attacks, and
Section~\ref{sec:conclusions} concludes. Proofs, parameter-selection guidance and supporting analyses are collected in the appendices.

\section{Overview and Taxonomy of Watermarking in LLMs}
\label{sec:overview}
\begin{figure*}[t]
  \centering
  \includegraphics[width=\textwidth]{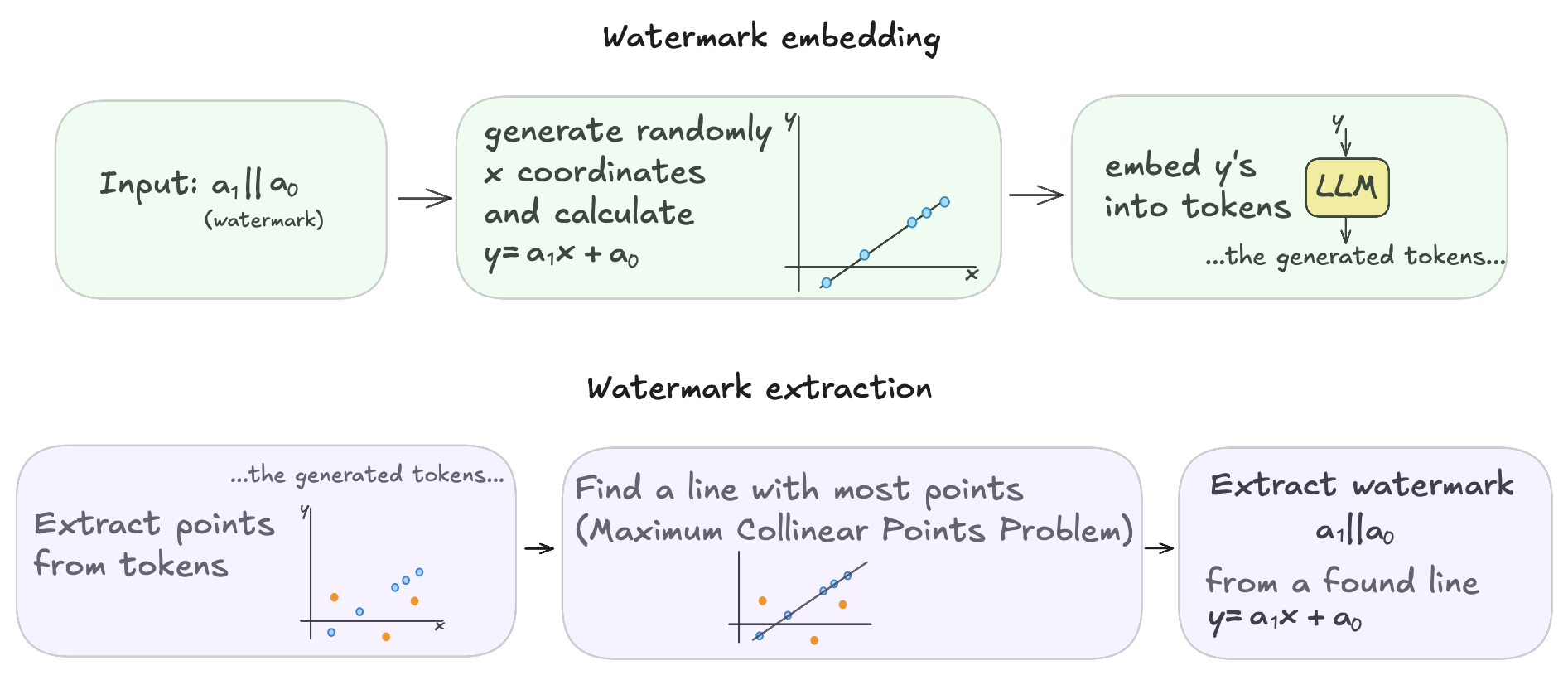}
  \caption{An overview of our watermarking scheme in its basic variant. A multi-bit watermark is represented by two numbers, \( a_1 \) and \( a_0 \), which are used as the coefficients of a linear function \( y = a_1x + a_0 \). We pseudo-randomly generate many \( x \)-coordinates and compute the corresponding \( y \)-coordinates. These \( y \)-values are then embedded into a generated text by biasing the token selection. In the extraction phase, we recover the embedded points from the generated text. Due to sampling variability during generation and/or adversarial tampering with the text, some of the recovered points may be incorrect (orange points). Our goal is to find the line that passes through the largest number of points, a problem formally known as the Maximum Collinear Points Problem. Once this line is found, its coefficients \( a_1 \) and \( a_0 \) constitute the extracted watermark.
}
  \label{fig:watermarking_scheme}
\end{figure*}

The primary motivation behind watermarking has long been the need to verify that a document or a digital artifact originates from its rightful source~\cite{cox2007}. Classical watermarking enables creators to prove ownership of multimedia content beyond a reasonable doubt. In the era of generative AI and LLMs, this requirement is more urgent than ever, due to the growing difficulty in distinguishing between human and machine generated texts. Watermarking for LLMs serves several essential purposes: (1) authenticating the origin of text, (2) preventing unauthorised reuse or plagiarism, (3) identifying misuse (e.g., misinformation, impersonation), and (4) establishing accountability in high-stakes domains (e.g., education, journalism, code generation).
Recent literature has explored a diverse range of strategies for embedding and detecting watermarks in LLM-generated content. Generally, these approaches can be divided into two categories:

\textit{Watermarking During Training} --- in this approach, the watermark is embedded by modifying the model training data or objectives, so that identifiable patterns are learned and can be later detected. In the work~\cite{Sander_2024}, the authors introduced the notion of text radioactivity. It allows to detect whether an LLM has memorised radioactive (marked) data points. TextMarker described in~\cite{Liu_2023} uses backdoor-based membership inference to embed private ownership information. The scheme~\cite{Cui_2025} injects plausible yet fictitious knowledge into training corpora to create semantically valid but watermarkable patterns.
    
\textit{Watermarking at the Logit Level} --- the output probability distribution (logits) is subtly perturbed before sampling tokens, embedding statistical patterns without altering the model weights. The work~\cite{kirchenbauer2023watermark} splits the collection of tokens into green and red subsets and biases the token sampling towards green tokens.  Wong et al.~\cite{wong_2025} describe their logits-to-text watermarking, which jointly optimises encoder-decoder architecture for high-fidelity and high-detectability watermarks. The GumbelSoft watermarking described in~\cite{Fu_2024} improves sampling diversity while maintaining robust detection by refining Gumbel-based logits control.

For a comprehensive survey of LLM watermarking techniques, readers are referred to the review by Liu et al.~\cite{Liu_2023review}. In this paper we are particularly interested in multi-bit watermarks, where embedded watermark consists of more information than just a binary signal (watermarked vs. non-watermarked). 

\subsection{Research Gaps}{
Despite the significant progress surveyed above, several important gaps remain in the state-of-the-art. First, the majority of watermarking schemes operating at the logit level embed only a binary signal — that is, they can determine whether a text is watermarked or not, but cannot recover a richer identifying payload 
~\cite{kirchenbauer2023watermark,aaronson2023watermark}.    
This limits their utility in practical attribution scenarios where it is necessary to identify, for example, a specific user, session, or model instance. Second, among the methods that do support multi-bit watermarks, extraction complexity often grows exponentially with watermark size, rendering them computationally infeasible for payloads of 32 bits or larger \cite{Fernandez2023ThreeBricks,Wang2024CodableWatermarking}.    
This presents a significant scalability challenge for real-world deployments. Third, existing schemes generally lack a robust algebraic recovery mechanism: once an adversary partially edits or corrupts the generated text, the embedded signal is typically degraded in a manner that is difficult to recover from in a principled way \cite{Yoo2024MultiBitWatermark}.    
Methods relying purely on statistical hypothesis testing are particularly vulnerable, as they do not distinguish between which portions of the text still carry valid watermark information and which have been tampered with.

The present work addresses all three gaps. By encoding the watermark as the coefficients of a polynomial over a Galois finite field $GF(2^n)$
and embedding a sequence of points $(x, f(x))$ into the generated tokens, we obtain an explicit multi-bit payload that can be recovered geometrically rather than purely statistically. The recovery problem reduces to finding the Maximum Collinear Points (MCP) among the extracted points, which can be solved efficiently in $\mathcal{O}(N^2)$ time regardless of watermark size, directly resolving the scalability limitation of prior approaches. Furthermore, the algebraic structure of the scheme means that even if a large fraction of embedded points are destroyed or corrupted by an adversary, the surviving points can still recover the original watermark, provided the line through them holds more points than any line formed by chance (Section~\ref{sec:security_evaluation}). Together, these properties position the proposed scheme as a complementary and practically advantageous alternative to existing multi-bit watermarking methods.
}

\subsection{Novelty}{
While the token embedding mechanism is based on logit biasing similarly to prior watermarking methods, the novelty of the proposed scheme lies in the representation and recovery of the watermark. In particular, the watermark is modeled as coefficients of algebraic structures over finite fields and recovered using a combination of Lagrange Interpolation and algorithms for the Maximum Collinear Points Problem. This formulation provides a different perspective on watermark extraction, where recovery is reduced to identifying geometric structures among noisy or partially corrupted points. An important consequence of this approach is that, for long texts, the watermark may still be recoverable even when only a small number of valid points survive adversarial edits (Section~\ref{sec:security_evaluation}).
}

\subsection{Contribution}
\subsubsection{Primary contribution}
We introduce an algebraic formulation of multi-bit watermarking for LLM-generated text. The payload is encoded as the coefficients of a polynomial over $GF(2^{n})$, and points $(x, f(x))$ of this polynomial are embedded in the generated tokens. In the linear case, which we instantiate and evaluate, recovery reduces to finding the Maximum Collinear Points among the extracted points, followed by Lagrange interpolation. Extraction then costs $\mathcal{O}(N^2)$ in the number $N$ of decoded blocks, independent of payload size. Like other logit-biasing methods, our scheme is applicable to multiple LLM architectures without model retraining.

\subsubsection{Supporting contributions}
We additionally provide: (i)~a probabilistic analysis of spurious collinearity among independent, uniformly random points in $\GF(2^{n})^2$ (Theorem~\ref{thm:rand_collinearity}) that guides the choice of field size, block length and detection threshold; (ii)~a lightweight \emph{bit-correction} step that trades runtime for recall; (iii)~a sliding-window extractor as a proof-of-concept remedy for desynchronization caused by insertions and deletions; and (iv)~an empirical evaluation of match rates on five open LLMs from the Llama and Mistral families, of false-positive rates on human-written text, and of robustness to token-level edits. Match rates vary across models and are lowest for instruction-tuned models at $32$ bits. We also outline extensions to multiple lines and higher-degree polynomials, left for empirical evaluation in future work. These items characterize and extend the primary construction; we do not claim full robustness to arbitrary edits.

\section{Preliminaries}
\label{sec:building_blocks}

\subsection{Lagrange Interpolation over $GF(2^n)$}
\label{subsec:lagrante_interpolation}

In the proposed watermarking scheme, we use Lagrange interpolation to recover a function from a set of distinct points.

Formally, Lagrange interpolation over a finite field $GF(2^n)$ constructs a unique polynomial $f(x) \in GF(2^n)[x]$ of degree at most \( t-1 \) that passes through a given set of \( t \) distinct points \( (x_1, y_1), \dots, (x_t, y_t) \), where all $x_i \in GF(2^n)$ and \( x_i \ne x_j \) for \( i \ne j \). The interpolating polynomial is given by:

\[
    f(x) = \sum_{j=1}^{t} y_j \cdot \ell_j(x),
\]
where each $\ell_j(x) \in GF(2^n)[x]$ is a \textit{Lagrange basis polynomial} defined as:
\[
    \ell_j(x) = \prod_{\substack{1 \leq m \leq t \\ m \ne j}} \frac{x - x_m}{x_j - x_m}.
\]

This form of interpolation is widely used in finite field cryptography, such as in Shamir’s secret sharing, error-correcting codes, and information dispersal schemes.

\subsection{Embedding Bits into Tokens}
\label{subsec:bits_tokens}

Kirchenbauer et al.~\cite{kirchenbauer2023watermark} introduced the concept of “Green” and “Red” tokens, where increasing the logit values of Green tokens leads to a strong preference for them --- a hallmark of watermarked text. In this paper, we use a similar method; however, instead of preferring only green tokens, we select either a green or red token depending on a bit value (green corresponds to 1, red to 0). In this way, we embed a binary vector (watermark) into the stream of generated tokens.

\section{Watermarking scheme based on Lagrange Interpolation}
\label{sec:watermarking_scheme}

In the proposed watermarking scheme, the watermark $W$ constitutes coefficients of the polynomial $f(x)$. We focus on the simplest nontrivial case where $f(x)$ is a linear function $f(x) = a_0 + a_1x \in GF(2^n)$, where the watermark is encoded as $W = (a_0 \| a_1) \in GF(2^n)$. We generate a set of points, each of which satisfies the polynomial $f(x)$. The $x$-coordinates are generated pseudo-randomly and the corresponding $y$-coordinates are embedded into tokens (as described in the previous section).

During the extraction, we solve Maximum Collinear Points Problem to identify the largest set of points, which are collinear (i.e., lie on the same line). Then we apply the Lagrange interpolation to those points to recover the original $f(x)$.  

A desirable feature of the proposed scheme is that even if only a fraction of embedded points survive tampering, they can still suffice to recover the watermark (Section~\ref{sec:security_evaluation}). The scheme is naturally resilient to partial modifications, making it suitable for adversarial environments where text may be edited with the intent to obscure provenance.

\begin{algorithm*}[h]
\caption{Watermark Embedding}
\label{alg:watermark_embedding}
\begin{algorithmic}[1]
\State \textbf{Input:} large language model $LLM$, vocabulary $\mathcal{V}$, $2n$-bit watermark $W=(a_0||a_1)$, bias $\delta$, prompt $s^{(-N_p)}, \dots, s^{(-1)}$, secret key $\mathcal{K}$, hash function $H$
\State \textbf{Output:} generated tokens $\mathcal{T}$ with embedded watermark $W$
    \hrule
    \State $f(x) = a_1x+a_0$ \Comment{specify a polynomial based on the watermark}
    \State $\mathcal{T}=s^{0} = LLM(s^{(-N_p)}, \dots, s^{(-1)})$ \Comment{generate the first token}
    \State $t=1$ \Comment{token counter}
    \While{not end of text generation}
        \State $x = H(s^{t-1}, \mathcal{K})$ \Comment{generate pseudo-random $x$-coordinate}
        \State $y=f(x)=(y_{0},y_{1},\ldots,y_{n-1})$ 
        \Comment{$y$ is an element of $GF(2^n$) subsequently mapped to $n$-bit vector for token embedding}

        \For{$j= 0$ to $n-1$} \Comment{a loop embedding $y$-coordinate into $n$ consecutive tokens}
            \State $seed = H(s^{t-1+j}, \mathcal{K})$
            \State Using $seed$, pseudo-randomly split vocabulary $\mathcal{V}$ into two sets $\mathcal{V}_0$, $\mathcal{V}_1$
            \State Apply $LLM$ to prior tokens $s^{(-N_p)} \dots s^{(t-1+j)}$ to get a logit vector $v$ over the vocabulary $\mathcal{V}$:
            \State $v = LLM(s^{(-N_p)}, \dots, s^{(t-1+j)})$
            \If{$y_{j}==0$}
                \State $\hat{v}[w]=v[w]+\delta$ for all $w\in \mathcal{V}_0$ \Comment{conditionally bias the logits}
            \Else
                \State $\hat{v}[w]=v[w]+\delta$ for all $w\in \mathcal{V}_1$
            \EndIf
            \State Sample the token $s^{(t+j)}$ from biased logits $\hat{v}$ and append it to $\mathcal{T}$
        \EndFor
        \State $t=t+n$
    \EndWhile
    \State \Return $\mathcal{T}$
\end{algorithmic}
\end{algorithm*}

\textit{Watermark Embedding:} Algorithm~\ref{alg:watermark_embedding} describes the process of embedding a multi-bit watermark into text generated by an LLM. The watermark is encoded as a linear polynomial $f(x)=a_1x+a_0$, where the coefficients $a_0$ and $a_1$ together form the watermark. The algorithm begins by initializing the text generation with a given prompt, and then iteratively embeds the watermark across the generated tokens. In each iteration, a pseudo-random $x$-coordinate is computed using a cryptographic hash function applied to the previous token and a secret key. This $x$-value is then used to evaluate the watermark polynomial, producing a corresponding $y$-value that represents the next set of bits to embed. For each bit in this $y$-value, the vocabulary is pseudo-randomly partitioned into two subsets, one corresponding to a bit value of 0 and the other to 1. The LLM generates a probability distribution (logits) over the vocabulary, which is then selectively biased toward the appropriate subset based on the bit being embedded. Then, a new token is sampled (greedy sampling) from the biased logits and added to a sequence of the generated tokens \( \cal T \). The process repeats for each $y$-value, with each one embedded across a fixed number of consecutive tokens, until the desired length of text is generated.

\begin{algorithm*}[h]
\caption{Watermark Extraction}
\label{alg:watermark_extraction}
\begin{algorithmic}[1]
\State \textbf{Input:} sequence of tokens $s^{0}, s^{1}, \dots, s^{N}$, size of the watermark (in bits) $2n$, vocabulary $\mathcal{V}$, secret key $\mathcal{K}$, hash function $H$, threshold $\tau$
\State \textbf{Output:} watermark $W=(a_0\|a_1)$ or $\emptyset$ (non-watermarked text)
    \hrule
    \State $\mathbb{P}=\varnothing$ \Comment{Create an empty list of points}
    \State $t=1$ \Comment{token counter}
    \While{$t+n-1\leq N$} \Comment{while a full block of $n$ tokens remains}
        \State $x = H(s^{t-1}, \mathcal{K})$ \Comment{extract $x$-coordinate}
        \For{$j= 0$ to $n-1$} \Comment{a loop extracting $y$-coordinate of a single point}
            \State $seed = H(s^{t-1+j}, \mathcal{K})$
            \State Using $seed$, pseudo-randomly split vocabulary $\mathcal{V}$ into two sets $\mathcal{V}_0$, $\mathcal{V}_1$
            \If{$s^{t+j}\in \mathcal{V}_0$}
                \State $y_j = 0$ \Comment{extract $y_j$ bit of $y$-coordinate}
            \Else
                \State $y_j = 1$
            \EndIf
        \EndFor
        \State $y=(y_{0},y_{1},\ldots,y_{n-1})$
        \State Append $(x,y)$ point to the list $\mathbb{P}$
        \State $t=t+n$
    \EndWhile
    \State Find the maximum collinear points (MCP) in $\mathbb{P}$
    \State Apply the Lagrange interpolation to maximum collinear points to obtain $f(x)=a_0+a_1x$
    \If{MCP $\geq \tau$}
        \State \Return $W=(a_0\|a_1)$
    \Else
        \State \Return $\emptyset$
    \EndIf
\end{algorithmic}
\end{algorithm*}

\textit{Watermark Extraction:} Algorithm~\ref{alg:watermark_extraction} describes the procedure for extracting a watermark from a sequence of tokens generated by a large language model (LLM), assuming the text was previously watermarked using the embedding method outlined in Algorithm~\ref{alg:watermark_embedding}. 

For each group of $n$ consecutive tokens, a pseudo-random $x$-coordinate is derived using a cryptographic hash function $H$ applied to the previous token and a secret key $K$. This process mimics the original embedding and ensures that the same $x$-values are regenerated during extraction.

For each such $x$-value, the algorithm recovers a corresponding $y$-value by processing $n$ consecutive tokens. Using a deterministic hash-based seed, the vocabulary is split into two subsets: $V_0$ and $V_1$. The current token is compared against these subsets to infer a single bit of the $y$-coordinate. If the token belongs to $V_0$, the corresponding bit is set to zero; otherwise, it is set to one. This procedure is repeated for $n$ tokens to fully reconstruct the $n$-bit $y$-coordinate.

Each recovered point $(x, y)$ is added to a collection $P$. After processing the entire token sequence in this manner, the algorithm attempts to identify the original watermark by solving the Maximum Collinear Points (MCP) problem on the set $P$. This step finds the largest subset of points that lie on a single straight line in the two-dimensional plane.

Once this maximal collinear subset is identified, Lagrange interpolation is applied to these points to reconstruct the polynomial $f(x)$. The coefficients of this polynomial, namely $a_0$ and $a_1$, are then returned as the extracted watermark $W = (a_0 \| a_1)$.

\subsection{Maximum Collinear Points Problem}

During the watermark embedding, a selected set of tokens is biased depending on the value of a watermark bit. However, due to the stochastic nature of sampling, there is no guarantee that a token from the desired set will be selected. As a result, the embedding process may introduce errors. Another potential source of errors is adversarial tampering with the text, such as adding or deleting tokens. Consequently, the verifier is faced with the task of identifying the genuine points (i.e., valid $x$ and $y$ coordinates), which can be formally framed as an instance of the Maximum Collinear Points (MCP) problem. The MCP problem is defined as follows:

{\it Instance}: A set $\mathbb S$ of $N$ points in the 2D plane, where each point
$p_i=(x_i,y_i)$; $i=1,2,\ldots,N$.

{\it Question}: What is the largest subset ${\mathbb S}' \subset {\mathbb S}$ of points that
are collinear (i.e., lie on a single line)?

There are many efficient algorithms that solve the problem. The basic ones are as follows:
\begin{itemize}

    \item Brute force approach — the algorithm checks all triplets for collinearity. Its time and space complexities are ${\cal O}(N^3)$ and ${\cal O}(1)$, respectively (see~\cite{Cormen_2022}).

    \item Hashing-based approach — the algorithm is similar to the sorting-based one, except that instead of sorting, it applies hashing to store slopes, reducing the time complexity to ${\cal O}(N^2)$ (see~\cite{Cormen_2022}).

\end{itemize}

Also there is a convex-hull approach~\cite{Preparata_1985} and sorting-based method~\cite{Cormen_2022}.

Let us consider the Hashing-based algorithm (see Algorithm~\ref{alg:hash_MCP}).
It determines the largest subset of collinear points among $N$ given points in a 2D plane. The algorithm iterates through each point, treating it as a reference. For every other point, it computes the slope relative to the reference and stores the slope count in a hash map. This hash map allows for quick lookups and updates, ensuring that collinear points are efficiently counted. The maximum count of any slope within the hash map, plus the reference point itself, gives the largest number of collinear points for that reference. By iterating through all points, the algorithm determines the global maximum. Note that vertical lines can be ignored as they are not expressed as a linear function $y=a_1x + a_0$. Blocks that follow the same token get the same $x$, and we keep all decoded points, including repeats. Ignoring vertical lines means that the search skips every pair of points with equal $x$. Two points with the same $x$ but different $y$ are therefore never counted on one line, and the line is always interpolated from points with distinct $x$. An exact copy of a point, however, is counted again whenever another point is the reference. All reported results use the decoded points as they are. A preliminary test on our Llama~2 runs suggests that removing copies before the search can improve recovery further, by $0.2$ to $1.2$ percentage points; a full evaluation of this refinement is left for future work.

\begin{algorithm*}[h]
\caption{Hashing-Based MCP (Maximum Collinear Points)}
\label{alg:hash_MCP}
\begin{algorithmic}[1]
    \State \textbf{Input:} A list of $N$ points ${\mathbb S} = (p_1, p_2, \dots, p_N)$ in $GF(2^n) \times GF(2^n)$
    \State \textbf{Output:} Maximum number of collinear points
    \hrule
    \State maxCount $\gets 1$
    \For{each point $p_i \in {\mathbb S}$}
        \State Create an empty hash map $H$
        \For{each point $p_j \in {\mathbb S}$ such that $p_j \neq p_i$}
            \State Compute slope $s = \frac{y_j - y_i}{x_j - x_i}$ {\Comment{{\textit{Pairs with $x_j = x_i$ are skipped}}}}
            \State Increment $H[s]$ by 1
        \EndFor
        \State maxCount $\gets \max($maxCount, $\max_{s \in H} (H[s] + 1))$
    \EndFor
    \State \Return maxCount
\end{algorithmic}
\end{algorithm*}

\subsubsection{Bit correction}
In practice, recovering the $n$-bit $y$ from $n$ consecutive tokens is noisy due to (i) stochastic sampling that can select a token outside the intended subset, (ii) local edits/insertions/deletions that perturb nearby bits. To increase the chance that a true point $(x,y)$ survives verification, we apply a lightweight bit-correction step at extraction, see Appendix~\ref{apdxsec:bit_correction} for details.

\section{Security Evaluation}
\label{sec:security_evaluation}

Once a text is generated, an adversary can insert, delete or substitute some tokens. Instead of analyzing each attack separately, we approach the problem in a more analytical way. We aim to determine the threshold beyond which a watermark cannot be extracted from the LLM-generated text. During extraction, we obtain $N=F+R$ number of points, where $F$ are valid points that lie on the original straight line $f(x)$ and $R$ are spurious points resulting from two sources: (1) noise introduced by the extraction process itself, and (2) modifications introduced by the adversary.
It is reasonable to assume that the $R$ points are randomly scattered on the plane $GF(2^n)\times GF(2^n)$. Clearly, the watermark extraction is successful if any subset of random points forming a straight line is no larger than $F$.

First, consider the following trivial facts about straight lines on the $GF(2^n)\times GF(2^n)$ plane:

\begin{itemize}
    \item Any two points uniquely determine a straight line. This means that there are $2^{2n}$ straight lines.
    \item Any two straight lines can either intersect at a single point or be parallel with no points in common.
    \item For each point, there are $2^n$ straight lines that contain that point.
\end{itemize}

The above properties show that for the number of random points $R$ larger than $2^n$, the correlation among the straight lines grows. Consequently, finding the precise probability distribution of the number of random points that belong to different straight lines is a challenging problem.

\begin{theorem}[Random Collinearity in $\mathrm{GF}(2^n)^2$]\label{thm:rand_collinearity}
Let $q = 2^n$. Let $S=\{p_1,\ldots,p_R\}$ be $R$ points sampled i.i.d.\ uniformly from $\mathrm{GF}(q)\times\mathrm{GF}(q)$, and consider only non-vertical lines of the form $y=ax+b$ with $a,b\in\mathrm{GF}(q)$. For any integer $k\ge 3$, the probability that there exists a non-vertical line containing at least $k$ points of $S$ satisfies

\[
    \Pr\big[\exists\, k\text{-collinear points}\big]
    \;\le\;
    \binom{R}{k}\, q^{-(k-2)}.
\]

A proof of the theorem is given in Appendix~\ref{app:proof-thm1}.
\end{theorem}

Theorem~\ref{thm:rand_collinearity} provides us with a tool to evaluate the probability that Algorithm~\ref{alg:watermark_extraction} extracts the watermark correctly. For a \emph{fixed} non-vertical line $\ell$, the probability that it contains exactly $k$ of the $R$ points is

\[
    \Pr\big[|\ell\cap S| = k\big] \;=\; \binom{R}{k}\left(\frac{1}{q}\right)^{k}\!\left(1-\frac{1}{q}\right)^{R-k}.
\]

\begin{example}
    Consider a toy example, where calculations are done in $GF(2^8)$, i.e. $n=8$. We have $F=4$ points lying on a straight line and $R=16$ random points. The expected number of straight lines that contain exactly four of the random points is

    \[
        \binom{16}{4}\left(\frac{1}{256}\right)^{2}\!\left(1-\frac{1}{256}\right)^{12}\;\approx\; 0.0265,
    \]

    This expected count is an upper bound on the probability that such a line exists. Consequently, the probability of correct recovery of the watermark is $\approx 0.9735$.
\end{example}

Reconstruction and detection should be distinguished. Two valid points with distinct $x$ suffice to reconstruct the line, but any two points are collinear, so such a line is no evidence of a watermark. Reliable detection needs more, because spurious points can also fall on a common line by chance. Algorithm~\ref{alg:watermark_extraction} therefore accepts the largest collinear set only when it holds at least $\tau$ points, where $\tau$ is chosen from the bound of Theorem~\ref{thm:rand_collinearity} for a target false-positive rate (between three and six for our payload sizes without bit correction, Table~\ref{tab:tau-selection}). The theorem treats spurious points as random, so we check it on human-written text in Appendix~\ref{app:fpr-empirical}: the resulting thresholds meet their targets at $\alpha=10^{-2}$ and $10^{-3}$, but not at $10^{-4}$ for three of the four payload sizes. Under adversarial edits, we measure rather than guarantee how many valid points survive (Section~\ref{sec:adversarial_attacks} and Appendix~\ref{app:sliding}).

\section{Experimental Results}
\label{sec:experimental_results}

We compared our approach with other top five methods. The results of the experiments are given in Table \ref{tab:comparison}. The experiment setting is identical to the one used in \cite{Qu:2024aa}. A watermarked text length is around 200 tokens. Texts were generated with Llama 2-7B model and prompts were taken from the OpenGen dataset \cite{opengendataset}. All experiments used a consistent logit bias of $\delta=6$. For longer watermarks (24 and 32 bits), we use the bit correction mechanism with $c=1$. All experiments were conducted on a single workstation equipped with an NVIDIA RTX 5090 GPU and an AMD EPYC 9354 32-Core Processor.
\begin{table}[htbp]
\centering
\caption{Comparison of methods on match rate(\%) for different watermark lengths (in bits).}
\label{tab:comparison}
\resizebox{\columnwidth}{!}{
\begin{tabular}{l|cccc}
    \toprule
    Watermark size &16&20&24&32 \\
    \midrule
    (Fernandez et al., 2023)  \cite{Fernandez2023ThreeBricks}           & 99.6 & 99.2 & 98.0 & NA \\
    (Wang et al., 2024)       \cite{Wang2024CodableWatermarking}        & 98.8 & 98.4 & 97.2 & NA \\
    (Yoo et al., 2024)        \cite{Yoo2024MultiBitWatermark}           & 73.6 & 49.2 & 30.4 & 8.4 \\
    (Cohen et al., 2025)      \cite{Cohen2025WatermarkingAdaptiveUsers} & 88.8 & 78.4 & 65.6 & 27.2 \\
    (Qu et al., 2024)         \cite{Qu:2024aa}                          & 98.0 & 97.6 & 96.0 & 94.0 \\
    Ours                                                                & 98.6 & 96.0 & 93.9 & 85.9 \\
    
    \bottomrule
\end{tabular}
}
\end{table}

For two baseline methods, the extraction time grows exponentially with watermark size and becomes infeasible in practice (denoted as ``NA''). The proposed method achieves competitive match rates across all evaluated payload sizes, although a larger performance gap emerges at the 32-bit setting~\cite{Qu:2024aa,Fernandez2023ThreeBricks}, while maintaining extraction times of only a fraction of a second.
The performance gap becomes more pronounced at the 32-bit watermark size. This behaviour is expected because larger payloads require larger algebraic structures and increase the probability that a block contains one or more bit errors, thereby reducing the probability of successful MCP reconstruction. Consequently, recovery performance degrades more rapidly with payload size than in some competing schemes that do not rely on geometric reconstruction. Nevertheless, the proposed framework remains scalable to larger payloads because watermark recovery depends on the successful reconstruction of only a small subset of valid blocks. As LLM generation quality improves and the effective bit-error rate decreases, the probability that a sufficient number of correct blocks survive is expected to increase, narrowing the performance gap at larger payload sizes.

Additionally, we evaluated five models from the \texttt{Llama} and \texttt{Mistral} families using two datasets \textit{Essays} \cite{Schuhmann2023EssaysWithInstructions} and \textit{HC3} \cite{guo-etal-2023-hc3}. The watermark's performance was tested across different watermark sizes and bit correction levels ($c$); 
all experiments used a consistent logit bias of $\delta=6$ and greedy decoding.   
Each experiment was set up to generate $208$ watermarked tokens using $100$ prompts per dataset. All the generations with fewer than $208$ tokens were filtered out.

The results of our experiments are shown in Table \ref{tab:llmw_results}. For both datasets and models, applying bit correction (Appendix \ref{apdxsec:bit_correction}) at the $c=1$ level usually shows some improvement in watermark recovery. However, there exists a delicate balance between the number of meaningful watermark blocks and random ones that are introduced by bit correction. As this threshold is crossed, the new blocks become detrimental to the MCP algorithm, as can be seen in the results for level $2$ bit correction.
The extraction time (MCP computation) with correction $c=1$ is still only a fraction of a second on a single PC (without any parallelization); at $c=2$ the expanded point set grows to $\sim1{,}600$ points and extraction takes $\sim7$ seconds (Appendix~\ref{subsec:bitcor_runtime}). 
\begin{table}[htb]
\centering
\caption{Match rates obtained with our method. Prompts were taken from a combination of the essays and HC3 datasets.}
\label{tab:llmw_results}
\resizebox{0.8\columnwidth}{!}{
\begin{tabular}{|c|l|c|c|c|}
    \hline
    $W$ size & Model & \multicolumn{3}{c|}{Bit Correction $c$} \\
    \hline
    & & 0 & 1 & 2 \\
    \hline
    \multirow{5}{*}[-4ex]{16}
    & \makecell[l]{Llama 2 \\ 7B}            & 98.3           & 94.2          & 76.9 \\
    \cline{2-5}
    & \makecell[l]{Llama 3.2 \\ 3B}          & 98.7           & 95.7          & 88.2 \\
    \cline{2-5}
    & \makecell[l]{Llama 3.2 \\ 3B Instruct} & 89.8           & 88.0          & 56.8 \\
    \cline{2-5}
    & \makecell[l]{Mistral \\ 7B}            & 98.2           & 94.4          & 84.4 \\
    \cline{2-5}
    & \makecell[l]{Mistral \\ 7B Instruct}   & 91.2           & 96.2          & 54.8 \\
    \hline
    \multirow{5}{*}[-4ex]{32}
    & \makecell[l]{Llama 2 \\ 7B}            & 83.9           & 91.1          & 60.2 \\
    \cline{2-5}
    & \makecell[l]{Llama 3.2 \\ 3B}          & 94.8           & 91.6          & 80.0 \\
    \cline{2-5}
    & \makecell[l]{Llama 3.2 \\ 3B Instruct} & 27.1           & 48.9          & 33.8 \\
    \cline{2-5}
    & \makecell[l]{Mistral \\ 7B}            & 90.2           & 95.2          & 65.7 \\
    \cline{2-5}
    & \makecell[l]{Mistral \\ 7B Instruct}   & 14.4           & 46.6          & 24.1 \\
    \hline
\end{tabular}
}
\end{table}

\subsection{Results Discussion}
\label{subsec:discussion}

The results in Tables~\ref{tab:comparison} and~\ref{tab:llmw_results} are best
read along two axes rather than one. The first is the familiar match rate at a
given payload size; on this axis the proposed scheme is competitive but not
uniformly best, trailing the strongest baseline~\cite{Qu:2024aa} by $8.1$ points
at $32$ bits. The second axis is the cost and the scaling behaviour of
extraction, and here the ordering is different. Our recovery procedure reduces to
a single MCP search over $N=\lfloor T/n\rfloor$ points, which is
$\mathcal{O}(N^2)$ \emph{irrespective of the payload size}: enlarging the
watermark changes the field $GF(2^n)$ in which the arithmetic is performed, not
the number of candidate structures that must be examined. The payload therefore
enters the runtime only through the width of the field elements.

Read as a function of payload size, Table~\ref{tab:comparison} separates the
methods into three groups. Two baselines~\cite{Fernandez2023ThreeBricks,
Wang2024CodableWatermarking} hold the highest rates over the range they cover,
losing only $1.6$ points between $16$ and $24$ bits, but report no $32$-bit
entry, their extraction becoming infeasible at that size as noted above. Two
others~\cite{Yoo2024MultiBitWatermark, Cohen2025WatermarkingAdaptiveUsers}
remain computable throughout but decay steeply, shedding $65.2$ and $61.6$
points respectively between $16$ and $32$ bits, so that neither exceeds $30\%$
at the largest payload. Only~\cite{Qu:2024aa} and the present scheme stay above
$85\%$ at $32$ bits, with degradations of $4.0$ and $12.7$ points from their
$16$-bit values. The proposed method thus sits in the third group while being
the only one whose extraction cost does not grow with the payload at all: for a
fixed segment of $T=208$ tokens, the number of points submitted to the MCP
search \emph{falls} from $N=26$ at $16$ bits to $N=13$ at $32$ bits
(Table~\ref{tab:tau-selection}), because a wider payload is carried by fewer,
wider blocks. Larger payloads are therefore cheaper to extract, not more
expensive --- the opposite of the trend that produces the ``NA'' entries.

The geometric formulation also shapes the robustness profile. Because the line can be recovered from whichever of its points survive, our scheme degrades gracefully as long as
these points form the largest collinear set (and, for detection, reach the threshold $\tau$), and it does not
require any particular point to survive. This is why substitution attacks, which
preserve block alignment, cost comparatively little
(Section~\ref{sec:adversarial_attacks}), and why the method tolerates a
substantial fraction of corrupted blocks at $16$ bits. The same property makes
the scheme unusually sensitive to insertions and deletions, which do not merely
corrupt points but destroy the alignment that produces them at all. Statistical
detectors in the style of~\cite{kirchenbauer2023watermark} have the opposite profile:
they accumulate evidence over a whole segment and are correspondingly
insensitive to where in the text the evidence appears, but they return a binary
decision rather than a payload. The two families are thus best viewed as
complementary, and a deployment that needs both a presence signal and an
attributable identifier could reasonably run them in parallel over the same
generation.

The weakest results in Table~\ref{tab:llmw_results} --- instruction-tuned models
at $32$ bits --- admit a single explanation that also indicates the remedy. In
the BSC view adopted throughout this paper, what matters is the probability
$P=(1-p)^n$ that an entire $n$-bit block is recovered without error, and this
quantity is exponentially sensitive to both the effective bit-error rate $p$ and
the block width $n$. Instruction tuning and RLHF raise $p$ by sharpening the
output distribution against the fixed logit bias $\delta$, and moving from $16$
to $32$ bits doubles $n$; the two effects compound. Neither is intrinsic to the
algebraic construction. Increasing $\delta$, lengthening the verified segment, or
distributing the payload over several lines of a smaller field
(Appendix~\ref{sec:multisecrets}) all reduce the exponent or increase the number
of blocks, and the bit-correction results in Table~\ref{tab:llmw_results} already
show a partial recovery of this kind at $c=1$, where match rates for the two
instruction-tuned models rise from $27.1\%$ and $14.4\%$ to $48.9\%$ and
$46.6\%$.

Taken together, these observations suggest a fairly clear deployment envelope.
The proposed scheme is attractive when an explicit, verifiable payload is
required, when extraction must be fast and its cost predictable as payloads grow,
and when the expected threat is bounded editing rather than wholesale rewriting.
It is less attractive for very short texts, where too few blocks are available to
overcome the detection threshold implied by Eq.~\eqref{eq:fpr-bound}, and for
adversaries who paraphrase or translate, against which no token-level scheme of
this family offers guarantees~\cite{pmlr-v235-zhang24o}. Within that
envelope, the parameters $(n,\tau,c)$ can be chosen analytically rather than by
search: $n$ from the payload and the available text length, $\tau$ from the
target false-positive rate via Eq.~\eqref{eq:tau-choice} (for $c=0$, with an empirical check at stringent targets, Appendix~\ref{app:fpr-empirical}), and $c$ from the
measured error rate of the target model. We regard this predictability, rather
than a marginal advantage in match rate, as the main practical benefit
of the algebraic formulation.

\section{Adversarial Attacks}
\label{sec:adversarial_attacks}

\begin{figure*}[t]
  \centering
  \includegraphics[width=\textwidth]{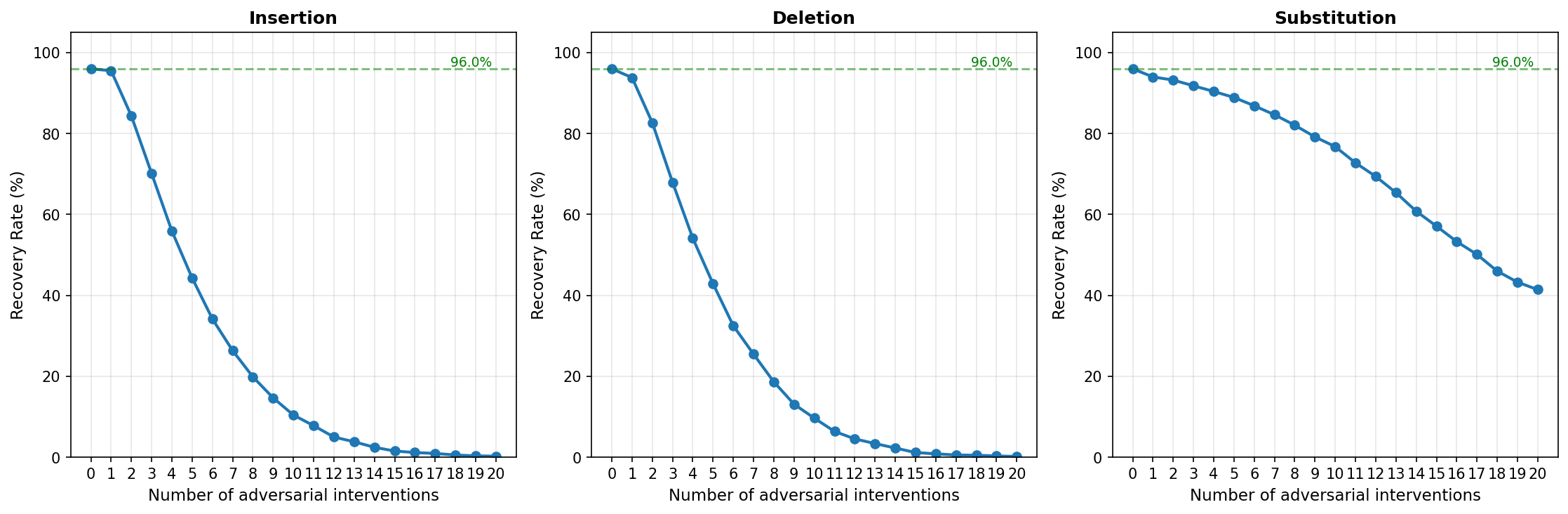}
  \caption{Effect of adversarial token-level interventions on watermark recovery.
Recovery rates are shown for substitution, insertion, and deletion attacks with a total budget of 10\% modified tokens and a text length of 200 tokens. The number of interventions controls the fragmentation of the attack (i.e., more interventions result in the attack being more evenly distributed throughout the text). Baseline performance without adversarial modifications is shown for reference.
}
  \label{fig:adversarial_attacks}
\end{figure*}

\subsection{Attack Model and Scope}

Following Kerckhoffs' principle, we assume the adversary knows the scheme and
all of its parameters and holds the generated text in full; only the key
$\mathcal{K}$ is secret. The adversary's capabilities are as follows.

\begin{itemize}
    \item \emph{Knowledge.} The extraction procedure and the parameters
    $n$, $\delta$, $\tau$ and $c$, hence the block width and the position of
    block boundaries relative to the start of the text. Not the vocabulary
    partitions $\mathcal{V}_0,\mathcal{V}_1$, which are pseudo-random under
    $\mathcal{K}$ and cannot be predicted without it.

    \item \emph{Actions.} Token-level substitutions, insertions and deletions,
    affecting at most a fixed fraction of the text ($10\%$ in our experiments).
    The budget is bounded because the edited text must remain useful to the
    adversary: one who is willing to discard the content can always defeat any
    watermark by discarding the text.

    \item \emph{Goal.} To prevent recovery of the embedded payload. We consider
    watermark \emph{removal}, not \emph{forgery} --- that is, the adversary
    seeks to destroy an existing payload rather than to embed a chosen one.

    \item \emph{Limits.} No access to $\mathcal{K}$, to the model weights, or to
    a detection oracle; in particular the adversary cannot test candidate edits
    against the verifier and iterate. We further assume that a key is not
    reused across enough generations for the vocabulary partitions to be
    estimated statistically.
\end{itemize}

Semantics-preserving rewriting ---
paraphrasing, translation and summarisation --- lies outside this model, as does
collusion between users holding differently watermarked texts. In our
experiments the interventions are placed at uniformly random positions; we
return to the stronger, alignment-aware adversary in the Limitations section.

\subsection{Substitution Attacks}

Substitution attacks replace existing tokens with alternative tokens. Crucially, substitutions preserve the overall token sequence length and therefore do not disrupt block alignment. As a result, each watermark block still corresponds to the intended group of tokens, although individual bits within a block may be corrupted.

This type of attack manifests primarily as \emph{bit noise} in the recovered $(x,y)$ points rather than structural damage. Since MCP-based recovery is inherently tolerant to a moderate number of spurious points, the watermark can still be recovered as long as enough valid points remain collinear. Consequently, substitution attacks are the least destructive class of interventions for our scheme, and the watermark exhibits partial robustness even without additional protection. For instance, Figure \ref{fig:adversarial_attacks} shows that the substitution of 9 out of 200 tokens causes the recovery rate to drop to 80\%.

\subsection{Insertion and Deletion Attacks (Desynchronization)}

Insertion and deletion attacks are significantly more damaging for our method because they alter the token stream length and thereby destroy block alignment. Even a single insertion or deletion can cause all subsequent watermark blocks to be misaligned, producing a cascade of invalid $(x,y)$ points during extraction. This effect is referred to as \emph{desynchronization}.

While MCP recovery can tolerate random noise, large-scale desynchronization rapidly increases the number of spurious points and may prevent the true collinear structure from dominating. Therefore, the basic scheme is inherently fragile with respect to insertion and deletion attacks and offers limited robustness against them.
Consider Figure \ref{fig:adversarial_attacks}: if we insert or delete 3 tokens, the recovery rate drops to approximately $75$\%.

To mitigate the synchronization problem, several possible approaches can be considered. The first is to modify the extraction algorithm so that it generates a point for every token by sliding the token sequence by one position instead of by $n$. Because both the abscissa and the vocabulary split are derived from the preceding token, such an extractor recovers every block that an edit does not physically overlap, irrespective of how far the block has been displaced.
Appendix~\ref{app:sliding} evaluates this variant on our generations in the case that is worst for the extractor of Algorithm~\ref{alg:watermark_extraction}, an insertion or deletion before the first block that shifts every block: that extractor recovers no watermark for any of the $6{,}507$ texts at any of the $27$ tested edit sizes, whereas the sliding extractor recovers the $20$-bit watermark in about $81\%$ of cases.

Another possible solution is to incorporate error-correcting coding into each block. This would make it possible to detect and reject blocks containing modified tokens, thereby improving robustness against synchronization errors. Finally, we could treat full stops ending sentences as a reset point preventing propagating the desynchronization. Each of the proposed solutions would be a trade-off between efficiency and robustness of the scheme.
We consider implementing the last two ideas as interesting future work.

\subsection{Experimental Evaluation}

We evaluate robustness under adversarial interventions affecting 10\% of the tokens, considering substitution, insertion, and deletion attacks. Figure~\ref{fig:adversarial_attacks} reports the watermark recovery rate as a function of the number of  \emph{adversarial interventions} applied to the text.

Across all attack types, recovery performance degrades gracefully as the number of interventions increases. As expected, substitution attacks exhibit the slowest degradation due to preserved synchronization whereas insertion and deletion attacks are more harmful. These results confirm that, while the proposed scheme is not fully robust against arbitrary editing, it can tolerate a non-trivial level of adversarial manipulation.

\section{Conclusion}\label{sec:conclusions}
This work introduces a novel watermarking scheme for LLM-generated texts. The core idea is to represent the watermark as coefficients of a straight line $f(x)$ over a finite field, and embed a sequence of points $(x_i,f(x_i))$ lying on this line into the LLM-generated text. A watermark is recoverable even after partial editing or adversarial manipulation. We demonstrate that the scheme is lightweight, straightforward to implement, 
and most suitable when decoding can be controlled and the text is not expected to undergo extensive editing.

The scheme is naturally extendable to support longer multi-bit watermarks. One straightforward approach is to embed multiple lines, each corresponding to a part of the watermark. We also propose the use of higher-degree polynomials instead of lines. In this approach, the watermark is encoded as the coefficients of a polynomial $f(x)$ and watermark points are sampled from this polynomial. 

Overall, the proposed watermarking framework offers a flexible balance between efficiency and security, and opens the door to new watermarking paradigms based on algebraic structures. Future work may focus on optimising variants with higher-degree polynomials, enhancing resilience against editing attacks, and exploring hybrid schemes combining lines and polynomials to balance robustness and performance.

Unlike many existing watermarking methods that primarily aim to detect the presence of a watermark through statistical tests, our approach embeds an explicit multi-bit payload represented by algebraic structures over finite fields. This enables watermark recovery through geometric reconstruction rather than solely hypothesis testing. The price paid for this capability is a stronger dependence on synchronization and the quality of recovered blocks. Consequently, the proposed method should be viewed as complementary to existing statistical watermarking schemes rather than as a direct replacement.

\section*{Limitations}
Our scheme assumes that enough $(x,f(x))$ points survive generation noise and light edits to dominate the MCP search; it is not designed to be robust against heavy paraphrasing, translation, aggressive summarization, or large insertion/deletion cascades that desynchronize blocks. 
This limitation is shared by many token-level watermarking approaches. Recent work has explored semantically invariant watermarking techniques that aim to preserve watermark detectability under paraphrasing and other meaning-preserving transformations~\cite{Liu_2024sir}. Investigating whether such ideas can be combined with the proposed MCP-based framework is an interesting direction for future work.

Our attack experiments place interventions at uniformly random positions. An adversary who knows the block width $n$ can do better: a single insertion near the beginning of the text desynchronises every subsequent block, whereas a randomly placed one costs on average half of them, so the worst case is harsher than the averages we report. Appendix~\ref{app:sliding} measures this worst case: with the edit placed before the first block, the extractor of Algorithm~\ref{alg:watermark_extraction} recovers no watermark, whereas the sliding-window extractor recovers about $81\%$. Separately, our model addresses watermark removal but not forgery; establishing that an adversary cannot embed a chosen payload, and thereby attribute text to an
innocent user, requires an unforgeability analysis that we leave to future work.

Theoretical bounds rely on independence assumptions and union bound arguments (and ignore vertical lines), which are informative but loose when $R$ grows relative to $q=2^n$. The empirical FPR measurements in Appendix~\ref{app:fpr-empirical}, on one corpus (HC3) and one tokenizer (Llama~2), show that real text departs from the independence and uniformity assumptions: the measured FPR exceeds the union bound at stringent thresholds (Table~\ref{tab:fpr-empirical}). 
Our implementation and experiments instantiate only linear $f$ and a single-line (or conceptual multi-line) payload. Higher-degree variants and ``top-$t$ lines'' composition are presented as design options rather than validated systems. 

Although our experiments cover diverse model families (Llama and Mistral), model scales (3B--7B), training paradigms (base and instruction-tuned), and three datasets (OpenGen, Essays, and HC3), a broader evaluation on recently released LLMs and additional datasets would provide a more comprehensive assessment of the generalizability of the proposed watermarking framework and is left for future work. Recovery also depends on the model: at $32$ bits and without bit correction, the two instruction-tuned models reach only $27.1\%$ and $14.4\%$, against $83.9$--$94.8\%$ for the base models (Table~\ref{tab:llmw_results}; see Section~\ref{subsec:discussion}).

\appendices
\section{Proof of Theorem~\thmoneref}
\label{app:proof-thm1}

\begingroup\def\IEEEproofindentspace{0pt}%
\begin{IEEEproof}[Proof of Theorem~\ref{thm:rand_collinearity}]
Let $q=2^n$ and let $S=\{p_1,\ldots,p_R\}$ be $R$ points sampled i.i.d.\ uniformly from $\GF(q)\times\GF(q)$. We consider only non-vertical lines of the form $y=ax+b$ with $a,b\in\GF(q)$.

\subsection{Counting non-vertical lines}
Each ordered pair $(a,b)\in\GF(q)\times\GF(q)$ determines a unique non-vertical line $\ell_{a,b}=\{(x,y): y=ax+b\}$, and distinct pairs give distinct lines. Hence the set $\mathcal{L}$ of non-vertical lines has cardinality $|\mathcal{L}|=q^2$.

\subsection{Membership probability and the per-line distribution}
Fix any $\ell\in\cL$. For a single random point $X$ uniformly distributed on $\GF(q)^2$, we have

\[
    \Pr[X\in \ell] \;=\; \frac{|\ell|}{|\GF(q)^2|} \;=\; \frac{q}{q^2} \;=\; \frac{1}{q}.
\]

Because $p_1,\ldots,p_R$ are sampled i.i.d., the random variable

\[
    N_\ell \;\defeq\; |\ell \cap S|
\]

is binomial: $N_\ell \sim \mathrm{Binomial}(R,1/q)$. In particular,

\[
    \Pr\!\big[N_\ell = k\big] \;=\; \binom{R}{k}\left(\frac{1}{q}\right)^k\!\left(1-\frac{1}{q}\right)^{R-k},
\]

which proves the ``per-line'' part of the theorem.

\subsection{Existence of a $k$-collinear subset on some non-vertical line (union bound)}
For a line $\ell\in\cL$ and a $k$-subset $T\subseteq [R]$ (indexing the points in $S$), let

\[
    E_{\ell,T}\;=\;\{\text{all points } \{p_i\}_{i\in T} \text{ lie on } \ell\}.
\]

Then $\Pr[E_{\ell,T}] = (1/q)^k$, since the $p_i$ are independent and each falls on $\ell$ with probability $1/q$. The event that there exist at least $k$ collinear points on some non-vertical line is contained in

\[
    \bigcup_{\ell\in\cL}\;\bigcup_{T\in \binom{[R]}{k}} E_{\ell,T}.
\]

Applying the union bound,

\begin{align*}
    &\Pr\!\big[\exists\, k\text{-collinear points on a non-vertical line}\big] \\
    &\le \sum_{\ell\in\cL}\;\sum_{T\in \binom{[R]}{k}} \Pr[E_{\ell,T}]
     = q^2\cdot \binom{R}{k}\cdot \left(\frac{1}{q}\right)^{k} \\
    &= \binom{R}{k}\, q^{-(k-2)}.
\end{align*}

This is exactly the claimed bound.
\end{IEEEproof}
\endgroup

\section{Tradeoff between Block Size and Probability of Success}
\label{apdxsec:prob_of_success}

In watermarking schemes that rely on identifying structures (such as collinear points) over finite fields like $GF(2^t)$, the designer is confronted with a tradeoff when selecting the block size $t$. This tradeoff arises when the underlying communication channel is modeled as a \emph{Binary Symmetric Channel} (BSC) with bit-flip (crossover) probability $p$, representing either a large language model (LLM) channel or another naturally error-prone medium.

Note that this assumption is a significant simplification of the error characteristics of LLMs, but it should be sufficient for our purposes.

On one hand, using a larger field (larger $t$) improves the geometric strength of the scheme—fewer collinear points are needed to identify a unique line, thus increasing robustness to adversarial actions. On the other hand, larger $t$ increases the likelihood that at least one bit in a block is corrupted, which degrades the reliability of receiving correct points.

To find a practical ``sweet spot'' between error resilience and field size, we examine the following result.

\begin{theorem}
    Let $t$ be the size (in bits) of a block transmitted over a Binary Symmetric Channel (BSC) with bit-flip probability $p$, and let $P = (1 - p)^t$ be the probability that a $t$-bit word is transmitted without error. Suppose $N$ such words are transmitted independently. Then, the probability that at least $n$ of them are received correctly is bounded below by:
    
    \begin{eqnarray*}
        \mathbb{P}[\text{at least } n \text{ correct blocks}] \\
        = \sum_{k = n}^{N} \binom{N}{k} P^k (1 - P)^{N - k} \geq L.
    \end{eqnarray*}
\end{theorem}

\begin{IEEEproof}
    Each $t$-bit block is correctly received with probability $P = (1 - p)^t$, and corrupted with probability $1 - P$. Let $X$ be the number of correctly received blocks among $N$ independent transmissions. Then $X$ follows the binomial distribution $X \sim \text{Binomial}(N, P)$. The probability that at least $n$ blocks are correct is:

\[
    \mathbb{P}[X \geq n] = \sum_{k = n}^{N} \binom{N}{k} P^k (1 - P)^{N - k}.
\]

This completes the proof.
\end{IEEEproof}

This theorem allows designers to tune the parameters of their schemes. A larger $t$ means fewer collinear points are needed to identify a line in $GF(2^t) \times GF(2^t)$, but the likelihood of receiving a block without error drops exponentially in $t$.

Hence, the selection of $t$ depends on the acceptable communication overhead (total number of transmitted tokens), the crossover probability $p$ (determined by channel quality), and the minimum number of correct blocks $n$ required to successfully identify the embedded structure (e.g., a watermark or secret line).

\begin{example}
Suppose the channel has a bit-flip probability of $p = 0.2$ (i.e., each bit is transmitted correctly with probability $0.8$). The table below shows the probability that an entire block of $t$ bits is received without error, i.e., $P = (1 - p)^t$, for various values of $t$.

    \begin{table}[h!]
    \centering
    \caption{Probability of a $t$-bit block being received without error over BSC$(p)$.}
    \resizebox{0.8\columnwidth}{!}{
    \begin{tabular}{|c|c|c|}
        \hline
        Block size $t$ & $P$ for $p = 0.1$ & $P$ for $p = 0.2$ \\
        \hline
        2 & 0.81    & 0.64    \\
        3 & 0.729   & 0.512   \\
        4 & 0.6561  & 0.4096  \\
        5 & 0.59049 & 0.32768 \\
        6 & 0.53144 & 0.26214 \\
        7 & 0.47830 & 0.20972 \\
        8 & 0.43047 & 0.16777 \\
        \hline
    \end{tabular}
    }
    \end{table}
    
    Now assume that $t = 6$ and $p = 0.2$, so the probability that a block is received error-free is $P = 0.26214$. We are interested in determining the minimum number of transmitted blocks $N$ such that the probability of receiving at least $n$ correct blocks is at least $L = 0.9$. The results are shown in Table \ref{table_transmitted_blocks}.
    
    \begin{table}[h!]
    \centering
    \caption{Minimum number of transmitted blocks needed to achieve 90\% confidence of receiving at least $n$ correct blocks with $t = 6$, $p = 0.2$.}
    \label{table_transmitted_blocks}
    \resizebox{0.8\columnwidth}{!}{
    \begin{tabular}{@{}|c|c|c|@{}}
        \hline
        $n$ (Correct) & Minimum $N$ & \#Tokens ($N \cdot t$) \\
        \hline
        4  & 15  & 90 \\
        8  & 34  & 204 \\
        12 & 50  & 300 \\
        16 & 66  & 396 \\
        32 & 132 & 792 \\
        64 & 258 & 1548 \\
        \hline
    \end{tabular}
    }
    \end{table}
\end{example}

\section{Practical parameter selection}
In our experiments, the block size $n$ was chosen from $\{8,10,12,16\}$ (corresponding to watermark sizes of $16$ to $32$ bits), as these values provide a reasonable compromise between the probability of block corruption and the false-positive rate of MCP detection. Larger fields reduce the probability of accidental collinearities but increase the likelihood that a block contains at least one erroneous bit. The detection threshold $\tau$ should be selected according to the target false-positive rate using Eq.~(\ref{eq:fpr-bound}), while the bit-correction radius $c$ represents a trade-off between recovery performance and computational cost. Empirical checks on human-written text (Appendix~\ref{app:fpr-empirical}, Table~\ref{tab:fpr-empirical}) indicate that Eq.~(\ref{eq:fpr-bound}) is reliable for choosing $\tau$ at loose targets ($\alpha\in\{10^{-2},10^{-3}\}$) but underestimates the false-positive rate at $\alpha=10^{-4}$, where the measured FPR exceeds the target for $n\in\{8,10,12\}$; for stringent targets, $\tau$ should be confirmed empirically rather than set from the bound alone. In practice, $c \in \{0,1,2\}$ was sufficient in our experiments; larger values rapidly increase the number of candidate points and the complexity of MCP recovery. Since different LLMs induce different effective bit-error rates (and can therefore be viewed as different Binary Symmetric Channels), the optimal choice of $(n,\tau,c)$ is model-dependent and should be tuned empirically for the target deployment scenario.
We note that the match rates reported in our experiments are obtained by exact coefficient recovery, i.e., by checking whether the recovered $(a_0, a_1)$ match the embedded watermark, and therefore do not require setting a detection threshold $\tau$; the threshold $\tau$ pertains to the presence-detection variant governed by Eq.~(\ref{eq:fpr-bound}).
Future work will include evaluating the proposed watermarking framework on newer generations of language models, such as Llama-4 and their successors, to further assess robustness, scalability, and generalization across evolving model architectures.

\section{Bit Correction}
\label{apdxsec:bit_correction}
When decoding a watermarked text, each $n$-bit block represents a $y$ coordinate in the Galois field $GF(2^{n})$; however, the errors introduced during encoding can corrupt the blocks, making them random and not useful for the MCP algorithm. To combat this issue, a $c$-bit error correction algorithm is employed that generates all possible block variants within $c$ bit flips of the recovered block. Mathematically, given a decoded $n$-bit sequence $y$, this produces $\binom{n}{c} = \frac{n!}{c!(n-c)!}$ alternative sequences, each representing a potential correction for up to $c$ simultaneous bit errors. During verification, all new variant $y$ coordinates and their corresponding $x$ values contribute to the MCP solver. This increases both the chances of having a correct sequence in the bag of the coordinates and the required computational power to calculate the new variants and solve the MCP problem.

Note that the bit-correction cost increases linearly with $n$ when $c=1$, whereas for $c=2$ the cost grows quadratically; the resulting extraction-time overhead is quantified in Appendix~\ref{subsec:bitcor_runtime}.

\subsection{Runtime overhead of bit correction}
\label{subsec:bitcor_runtime}

To quantify the computational overhead introduced by bit correction, we measured the per-phase extraction time---fixed-block decoding, candidate generation, and MCP solving---for the 32-bit setting ($n=16$) at $c\in\{0,1,2\}$, reusing the generated texts from the $n=16$ OpenGen experiment (500 prompts, Llama-2-7b, on the same workstation as the rest of the paper). Table~\ref{tab:bitcor_runtime} reports the results.

\begin{table}[htbp]
\centering
\caption{Per-phase extraction time for the 32-bit setting ($n=16$), averaged over 500 OpenGen prompts (Llama-2-7b) on the paper's workstation. $M$ is the number of points passed to the MCP solver after expanding each of the $N\approx12$ decoded blocks by $\sum_{i=0}^{c}\binom{n}{i}$ variants within Hamming radius $c$. Total includes a common fixed-block decode step ($\sim34$\,ms, or $\sim93$\,ms at $c=0$ due to one-time vocabulary-split setup); the bit-correction-specific overhead is the candidate generation plus the inflated MCP solve.}
\label{tab:bitcor_runtime}
\resizebox{\columnwidth}{!}{
\begin{tabular}{|c|c|c|c|c|}
\hline
$c$ & Expanded points $M$ & Candidate gen.\ (ms) & MCP solve (ms) & Total (ms) \\
\hline
0 & 12 & $<0.01$ & 0.5 & 94 \\
1 & 198 & 0.8 & 109 & 144 \\
2 & 1{,}595 & 7.2 & 7{,}290 & 7{,}332 \\
\hline
\end{tabular}
}
\end{table}

The candidate-generation step itself is negligible: ${<}0.01$\,ms at $c=0$, $0.8$\,ms at $c=1$, and $7.2$\,ms at $c=2$. The dominant cost is the MCP solve on the expanded point set. Bit correction inflates the point set from $M\approx12$ ($c=0$) to $M\approx198$ ($c=1$) to $M\approx1{,}595$ ($c=2$), and since the hashing-based MCP solver runs in $\mathcal{O}(M^2)$, the solve time grows from $0.5$\,ms to $109$\,ms to $7{,}290$\,ms. Total extraction time is therefore $\sim94$\,ms ($c=0$), $\sim144$\,ms ($c=1$), and $\sim7.3$\,s ($c=2$). The ``fraction of a second'' extraction claim thus holds for $c=0$ and $c=1$, but not for $c=2$ at $n=16$, where the $\binom{n}{2}$-fold expansion of the point set makes the quadratic MCP cost dominant.

In practice, $c=1$ is the recommended operating point: it is fast (${\sim}144$\,ms) and achieves the best recovery (Table~\ref{tab:llmw_results}). Increasing to $c=2$ at $n=16$ is not worthwhile---it costs ${\sim}50\times$ more time and recovers less, as the additional spurious points become detrimental to the MCP algorithm (Table~\ref{tab:llmw_results}). This is consistent with the observation that level-2 correction crosses the balance between meaningful and random blocks. The fixed-window and early-stopping optimizations for MCP detection can bound $M$ further; for $c=2$, however, restoring sub-second extraction would come at a robustness cost.

\section{Sliding-Window Extraction under Desynchronization (Proof-of-Concept)}
\label{app:sliding}

Section~\ref{sec:adversarial_attacks} suggested, as a remedy for insertion and
deletion attacks, an extractor that advances the block pointer by one token
instead of by $n$. This appendix specifies that extractor, bounds its cost, and
evaluates it on the \texttt{Llama~2~7B} generations behind
Figure~\ref{fig:adversarial_attacks}, starting from the edit that is most
destructive for the aligned extractor of
Algorithm~\ref{alg:watermark_extraction}.

\subsection{Mechanism}
In Algorithm~\ref{alg:watermark_extraction} both
$x=H(s^{t-1},\mathcal{K})$ and the vocabulary seed $H(s^{t-1+j},\mathcal{K})$ are
derived from the \emph{preceding} token. An insertion or deletion therefore
shifts a suffix of the token stream without invalidating the tokens themselves:
a token assigned to $\mathcal{V}_1$ at embedding is still assigned to
$\mathcal{V}_1$ at extraction, because the token it is hashed against has moved
with it. What an edit destroys is only the \emph{grouping} of the recovered bit
stream into $n$-bit blocks. Reading blocks at the fixed offsets $0,n,2n,\dots$,
every block after the first edit is read across a boundary with an abscissa
hashed from the wrong token, and contributes a point that is effectively uniform
on $\GF(2^n)^2$ --- hence the cascade of Section~\ref{sec:adversarial_attacks},
in which one inserted token costs half the payload.

The sliding-window extractor drops the assumption that the grouping offset is
known, emitting one candidate point per token position $o$,
\begin{equation}
\label{eq:sliding-point}
    \big(x_o,\,y_o\big)=\big(H(s^{o-1},\mathcal{K}),\; b_o b_{o+1}\cdots b_{o+n-1}\big),
\end{equation}
where $b_t\in\{0,1\}$ records whether $s^{t}$ falls in $\mathcal{V}_0$ or
$\mathcal{V}_1$ under the seed $H(s^{t-1},\mathcal{K})$. This is the single
change $t \gets t+1$ in place of $t \gets t+n$. A block embedded at position
$i\cdot n$ and displaced by a net drift $\delta_i$ is re-discovered at offset
$i\cdot n+\delta_i$, whatever $\delta_i$ is, so only blocks that an edit
physically overlaps are lost. Embedding
(Algorithm~\ref{alg:watermark_embedding}) is untouched, so the variant applies
to already-watermarked texts.

\subsection{Cost}
The per-token hashing and vocabulary splitting, which dominate extraction, are
unchanged: Algorithm~\ref{alg:watermark_extraction} already evaluates $H$ once
per token, so the bit stream $b_0b_1\cdots$ is available after a single
$\mathcal{O}(T)$ pass and the sliding extractor merely reads different windows
out of it. Only the MCP search grows. The number of points rises from
$N_{\mathrm{aligned}}=\lfloor T/n\rfloor$ to $N_{\mathrm{sliding}}=T'-n$, where
$T'$ is the edited length, so with MCP quadratic in $N$
(Algorithm~\ref{alg:hash_MCP}) the search cost grows by
\begin{equation}
\label{eq:cost-ratio}
    \big(N_{\mathrm{sliding}}/N_{\mathrm{aligned}}\big)^{2}\approx n^{2},
\end{equation}
about $59\times$ for $n=8$ and $235\times$ for $n=16$ at $T=200$. In absolute
terms this is minor: Table~\ref{tab:sw-cost} reports $1.1$~ms per detection
against $0.05$~ms for the aligned extractor, and even an $800$-token segment
($870$ points) is solved in $23$--$27$~ms.

The substantive price is the false-positive budget. Presenting $\approx n$ times
as many random points to the detector forces the threshold of
Eq.~\eqref{eq:tau-choice} upward: at $T=200$ and $\alpha=10^{-2}$, from $5$ to
$10$ for $n=8$ and from $3$ to $5$ for $n=16$. Recovery requires
$\mathbb{E}[F]\ge\tau$, and with $P=(1-p)^n$ the probability that a block is
error-free and $k$ edit runs of $\ell$ tokens placed uniformly at random, the
expected number of valid points is $\approx P N/(k+1)$ for the aligned extractor,
which keeps only the blocks preceding the \emph{first} edit, against
$\approx P\,(N-k(\ell+n)/n)$ for the sliding one, which loses only the blocks an
edit overlaps.

\begin{table}[htbp]
\centering
\caption{Cost for a $T=200$ token segment. ``MCP cost'' is the relative number
of point pairs inspected; ``Time'' is the measured wall-clock time of one MCP
call, averaged over $200$ random instances, in an unoptimised NumPy
implementation on a single CPU core. $\tau$ is the smallest threshold satisfying
Eq.~\eqref{eq:tau-choice} at $\alpha=10^{-2}$, cf.\ Table~\ref{tab:tau-selection}.}
\label{tab:sw-cost}
\resizebox{\columnwidth}{!}{
\begin{tabular}{|c|c|l|c|c|c|c|}
\hline
$W$ (bits) & $n$ & Extractor & Points $N$ & MCP cost & Time (ms) & $\tau$ \\
\hline
16 & 8  & aligned & 25  & $1\times$   & 0.05 & 5  \\
16 & 8  & sliding & 192 & $59\times$  & 1.13 & 10 \\
32 & 16 & aligned & 12  & $1\times$   & 0.03 & 3  \\
32 & 16 & sliding & 184 & $235\times$ & 1.04 & 5  \\
\hline
\end{tabular}
}
\end{table}

\subsection{Experiment}
We attack the $6{,}507$ \texttt{Llama~2~7B} generations with $n=10$ (a $20$-bit
watermark) that underlie Figure~\ref{fig:adversarial_attacks}, editing their
stored token sequences directly. Each attacked text is read by the aligned
extractor of Algorithm~\ref{alg:watermark_extraction}, which cuts blocks at the
fixed offsets $0,n,2n,\dots$, and by the sliding extractor; both read the same
green/red bit stream, use the same MCP solver and have bit correction disabled
($c=0$). As in Table~\ref{tab:llmw_results}, an extractor \emph{recovers} the
watermark when the line it returns equals $(a_1,a_0)$. It \emph{detects} the
watermark when, in addition, that line holds at least $\tau$ points, where
$\tau$ is the smallest threshold satisfying Eq.~\eqref{eq:tau-choice} at
$\alpha=10^{-2}$ for the number of points the extractor produces: $\tau=4$ for
the $17$--$22$ aligned blocks and $\tau=7$ for the $162$--$220$ sliding windows.
We did not measure the sliding extractor's false-positive rate on unwatermarked
text, so its threshold rests on the union bound alone.

The first attack is the one most destructive for the aligned extractor: $k$
random tokens are inserted before the first token, or the first $k$ tokens are
deleted, for $k=1,\dots,29$. Every block then moves by $k$ positions. We never
use an edit whose length is a multiple of the block length $n$, so $k=10$ and
$k=20$ are excluded: such a shift would put every later block back on the fixed
grid and re-synchronise the aligned extractor. The second attack shows how the
damage depends on where the edit lands: a single insertion or deletion of
$k\in\{1,5,15,25\}$ tokens at $0\%$, $25\%$, $50\%$, $75\%$ or $100\%$ of the
text, where an insertion at $100\%$ is appended and a deletion at $100\%$
removes the last $k$ tokens. Inserted tokens are drawn uniformly from the
vocabulary, as in Figure~\ref{fig:adversarial_attacks}, and every text uses a
fixed random seed.

\subsection{Results}
Figure~\ref{fig:sw-front} shows the edits at the start of the text. The aligned
extractor fails without exception: it recovers no watermark for any of the
$6{,}507$ texts, at any of the $27$ edit sizes, after an insertion or a deletion,
because each of its blocks now straddles two embedded blocks and on average only
$0.02$ of its points per text lie on the embedded line. Averaged over $k$, the
sliding extractor recovers the watermark in $80.6\%$ of the texts after an
insertion and in $81.3\%$ after a deletion, and detects it in $70.2\%$ and
$65.5\%$, against $86.9\%$ and $75.0\%$ without an attack. An insertion costs it
only the first block, whose $x$-coordinate is now hashed from an inserted token, so
its rates barely depend on $k$. A deletion removes the blocks it overlaps, one
more block for every further $10$ deleted tokens, which produces the steps in the
right panel.

Figure~\ref{fig:sw-position} moves a single edit through the text. The aligned
extractor keeps every block before the edit, so its recovery rises as the edit
moves later: averaged over $k$, from $0\%$ at the start to about $34\%$, $76\%$
and $91\%$ at $25\%$, $50\%$ and $75\%$ of the text, and to $95$--$96\%$ at the
end. The sliding extractor stays between $81\%$ and $86\%$ wherever the edit
lands. The curves cross between $50\%$ and $75\%$: an edit in the first half of
the text favours the sliding extractor, an edit in the last quarter the aligned
one.

The sliding extractor has a cost on unedited text: it recovers the watermark in
$86.9\%$ of the texts against $96.0\%$ for the aligned extractor. To keep the
higher rate, we use it only as a fallback. We run the aligned extractor first,
and only if no line reaches its threshold $\tau=4$ do we pass the text to the
sliding extractor, in case an edit near the beginning has shifted the block grid.

\begin{figure}[htbp]
\centering
\includegraphics[width=\columnwidth]{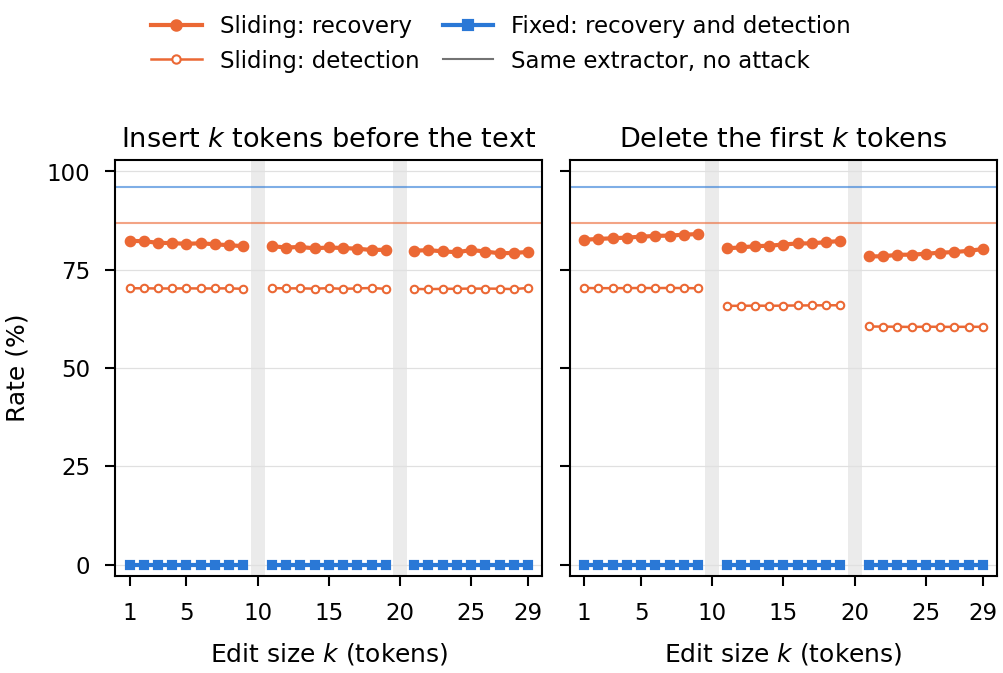}
\caption{Edits at the start of the text: $k$ random tokens inserted before the
first token (left) or the first $k$ tokens deleted (right), on $6{,}507$
generations with $n=10$. The aligned (fixed-block) extractor recovers and detects
no watermark at any $k$. Thin lines mark each extractor's recovery without an
attack; the shaded sizes $k=10$ and $k=20$ are excluded because they would
re-synchronise the fixed blocks.}
\label{fig:sw-front}
\end{figure}

\begin{figure}[htbp]
\centering
\includegraphics[width=\columnwidth]{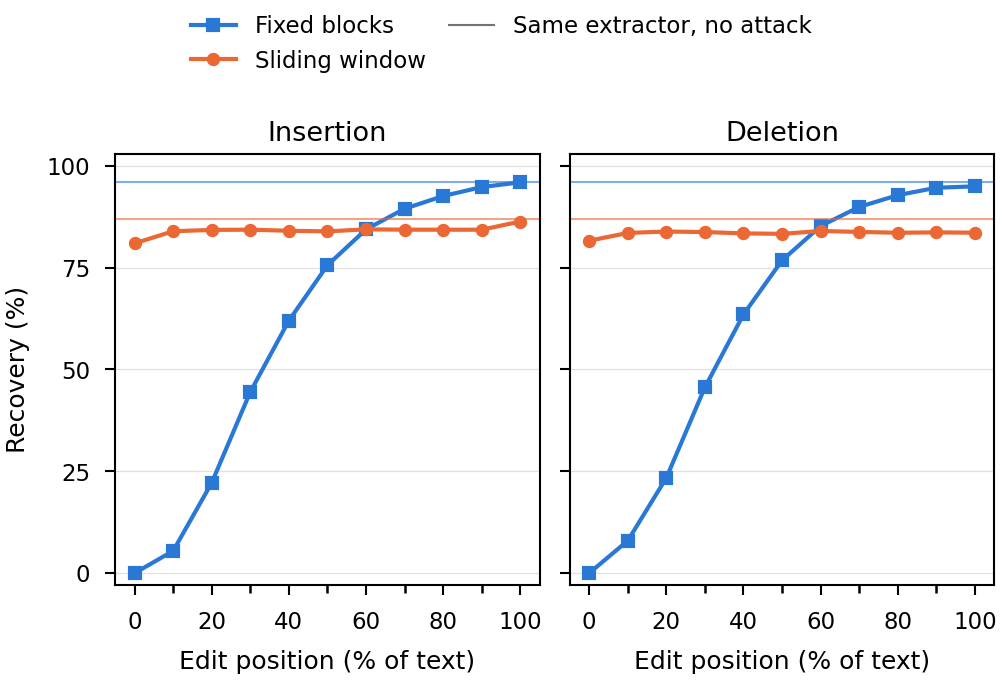}
\caption{Recovery after a single insertion (left) or deletion (right) at
increasing positions in the text, on $6{,}507$ generations with $n=10$. Each
point is the mean over the edit sizes $k\in\{1,5,15,25\}$, which differ by at
most $4.4$ percentage points at any position. Thin lines mark each extractor's
recovery without an attack.}
\label{fig:sw-position}
\end{figure}

\subsection{Practical notes and limitations}
The $n^2$ factor in \eqref{eq:cost-ratio} can be reduced if the verifier assumes
a bound $D$ on the net drift: scanning only offsets $i\cdot n+\delta$ with
$|\delta|\le D$ costs $(2D+1)^2$ instead of $n^2$, with a smaller increase in
$\tau$, and coincides with the full sliding window once $D\ge n/2$. Sliding
should not be combined with bit correction
(Appendix~\ref{apdxsec:bit_correction}), which multiplies the candidate set by a
further $n+1$ at $c=1$; the two address the same shortage of valid points by
opposite means.
Finally, the experiments above use real generations but a simple adversary,
who inserts or deletes random tokens in one place and does not paraphrase; the
sliding extractor's empirical false-positive rate remains future work.

\section{Watermarking with Multiple Lines and High Degree Polynomials}\label{sec:multisecrets}
\subsection{Multiple Lines}
Using a single line $f(x)=ax+b\in GF(2^n)$, it is possible to obtain the watermark $W=(a\| b)\in \{0,1\}^{2n}$.
Further expansion of its length is possible by embedding $t$ distinct lines instead of just one. In this case, our MCP algorithm requires a minor adjustment to track the top $t$ straight lines that yield the highest point counts. This results in a collection of lines $\{f_i(x)=a_ix+b_i \;| \; i=0,\ldots,t-1\}$ and their corresponding watermarks $\{W_i=(a_i \|b_i)\; | \; i=0,\ldots,t-1\}$. To form a single $2tn$-bit watermark, we must determine an appropriate order for concatenation. 

The method of concatenation depends on how the watermark is generated by the owner. The owner may either (1) generate individual pieces $W_i$ and concatenate them freely if the watermark is a random string created at the time of watermark embedding, or (2) follow a predetermined order if the watermark already exists and its structure is imposed by the application.
In the case where the author is free to create a new watermark, a simple and deterministic method can be used:

\begin{itemize}
    \item For each watermark component $W_i$, where $i = 0, \ldots, t-1$, compute its digest: $d_i = H(W_i)$, where $H()$ is a cryptographic hash function.
    
    \item Sort the digests ${d_0, \ldots, d_{t-1}}$ in increasing order and reorder the corresponding watermark components $W_i$ accordingly.
    
    \item Concatenate the reordered $W_i$ values into a single sequence.
\end{itemize}

If the watermark is fixed in advance, we can encode the order of components by indexing lines based on the least significant bits (LSB) of their $x$-coordinates. For simplicity, assume there are $t = 2^e$ lines, where $e \in {\mathbb N}$. Then:

\begin{itemize}
    \item Any point whose $x$-coordinate has an $e$-bit LSB equal to $d$ belongs to the line $f_d(x)$.

    \item The lines are ordered by increasing values of $d$, i.e., $f_0(x), f_1(x), \ldots, f_{2^e - 1}(x)$.
    
    \item The watermark $W=W_0\|W_1\|\ldots\|W_{2^e-1}$.
\end{itemize}

This idea is illustrated in the following simple case.

\begin{example}
Consider a watermarking scheme that consists of four integers $(a,b,c,d)$, each from some $GF(2^n)$, where $f_0(x)=ax+b$ and $f_1(x)=cx+d$. In the watermark embedding phase, all x-coordinates of the first line have their LSB equal to zero, while all points of the second line have their LSB equal to one.
During the recovery phase, we determine the two lines and have at least two points for each line. We then check their LSBs to determine the order of the lines and reconstruct the final watermark
$W=(a\|b\|c\|d)$.
\end{example}

\subsection{High Degree Polynomials}
An obvious generalisation is a choice of higher degree polynomials for our $f(x)$ instead of straight lines. This means that we deal with the polynomial over $GF(2^n)$ of the degree $t-1$ of the following form

\[
	f(x)=a_0+a_1x+,\ldots,a_{t-1}x^{t-1} \mbox{ where } t\geq 3
\]

Points $\{(x_i,f(x_i))\; | \; i=1,\ldots,N\}$ are translated into appropriate watermarking bits. Clearly, to determine the polynomial $f(x)$, it is enough to know $t$ points. We can apply the well-known Lagrange Interpolation to reconstruct $f(x)$. The watermark $W=(a_0\;\|\;a_1\;\| \; \ldots \;\|\;a_{t-1})$. The problem of recovery, however, becomes interesting when some points are corrupted by our adversary. The verifier in this case faces a task to identify a polynomial of degree $(t-1)$, which contains the maximum number of points. This is to say that we face the Maximum Co-Polynomial Points (MCPP) problem.

The MCPP problem is defined as follows:

{\it Instance}: A set $\mathbb S$ of $N$ points in the 2D plane, where each point
$p_i=(x_i,y_i)$; $i=1,2,\ldots,N$.

{\it Question}: What is the largest subset ${\mathbb S}' \subset {\mathbb S}$ of points that
lie on a single univariate polynomial $f(x)$ of degree at most $(t-1)$?

Some basic algorithms for solving the problem are: brute force polynomial interpolation ~\cite{Berrut2004}, Hough Transform for Polynomials ~\cite{Ballard_1981} or Gr\"{o}bner Basis ~\cite{Sauer_2006}.

Finding a polynomial $f(x)$ of degree $(t-1)$ from $t$ points becomes increasingly challenging as $t$ grows. However, our scenario is slightly different. If the watermarks remain unaltered by the adversary, any collection of $t$ such points will correctly determine the polynomial. In contrast, if the text has been modified, the points extracted from the tampered sections will not lie on the original polynomial $f(x)$.

The verifier must distinguish between two types of points: good ones that lie on $f(x)$ and bad ones, which are scattered randomly in the plane $GF^2(2^n)$. The following theorem evaluates the probability of success when the verifier randomly selects $t$ points, hoping they all lie on $f(x)$.

\begin{theorem}
    Given a collection ${\mathbb P}$ of all watermarks/points recovered from LLM-generated text, where ${\mathbb P} = {\mathbb F} \cup {\mathbb A}$, and ${\mathbb F}$ are points lying on $f(x)$ while ${\mathbb A}$ consists of random points introduced by the adversary. Suppose the verifier selects $t$ points at random, computes $\widetilde{f(x)}$, and checks whether other points are consistent with it. Then, the expected number of trials required to succeed (i.e., to find at least one set of $t$ points entirely from ${\mathbb F}$) is
    
    \[
    	\left(\frac{\ell}{f}\right)^{t},
    \]
    
    where $\ell=\#{\mathbb P}$ and $f=\#{\mathbb F}$.
\end{theorem}

\begin{IEEEproof}
    The probability of success in a single trial is $p = (\nicefrac{f}{\ell})^{t}$. If a trial fails, the next has probability $(1 - p)p$ of being successful. More generally, the probability that the $k$-th trial is the first success is $(1 - p)^{k-1}p$. This defines the geometric probability distribution $P(X = k) = (1 - p)^{k-1}p$, whose expected value is $\nicefrac{1}{p}$.
\end{IEEEproof}

\begin{example}
    Let $\ell = 1000$, and let the number of good points be $f = 200$. If the watermark is defined by a polynomial of degree $t-1 = 3$, then the expected number of trials is $(\nicefrac{1000}{200})^4 = 625$.
\end{example}

\section{Greedy sampling and lexical diversity}

Most LLM deployments almost universally employ top-p (nucleus) sampling, or top-k sampling, which inject stochasticity into token selection. In our case, this extra randomness is introduced by biasing token logits (before any sampling), which is pseudorandom due to a secret key and a hash function. Therefore, we do not need an extra layer source of lexical diversity and by applying greedy sampling (to such biased logits), we keep the effective bit-error rate low. We stress that the claim made in this section concerns \emph{lexical diversity}---the variety of vocabulary and n-grams used across a generation---and aggregate repetition density (as captured by the Trigram Repetition Rate); it is not a claim that repetition loops are eliminated. Repetition loops are a related but distinct degeneration phenomenon that we address separately in Appendix~\ref{app:repetition_loops}. To support our claim, we measure lexical diversity for the baseline LLM and our scheme. The non-watermarked baseline uses the default decoding configuration of \texttt{Llama~2~7B}, nucleus sampling ($\texttt{temperature}=0.6$, $\texttt{top\_p}=0.9$), with no logit bias, while the watermarked output uses greedy decoding ($\texttt{temperature}=0$) with a logit bias of $\delta = 6$. 
We measure word-level lexical diversity with three metrics: Distinct-1 (unique unigrams over total unigrams), Distinct-2 (unique bigrams over total bigrams), and Trigram Repetition Rate ($1 - \text{unique}/\text{total}$ trigrams). Higher Distinct scores indicate more varied vocabulary; lower Trigram Repetition Rate indicates less repetition. Figure \ref{fig:lexical_diveristy} shows the distributions for \texttt{Llama~2~7B} ($\approx$7,000 texts per variant). For all three metrics our scheme provides more diverse and less repetitive texts than the baseline. Since the two variants differ in both the decoding method and the logit bias, this difference cannot be attributed to the watermark alone; the comparison shows that greedy decoding combined with our watermark does not reduce lexical diversity relative to the model's default sampling.

\begin{figure*}[t]
  \centering
  \includegraphics[width=\textwidth]{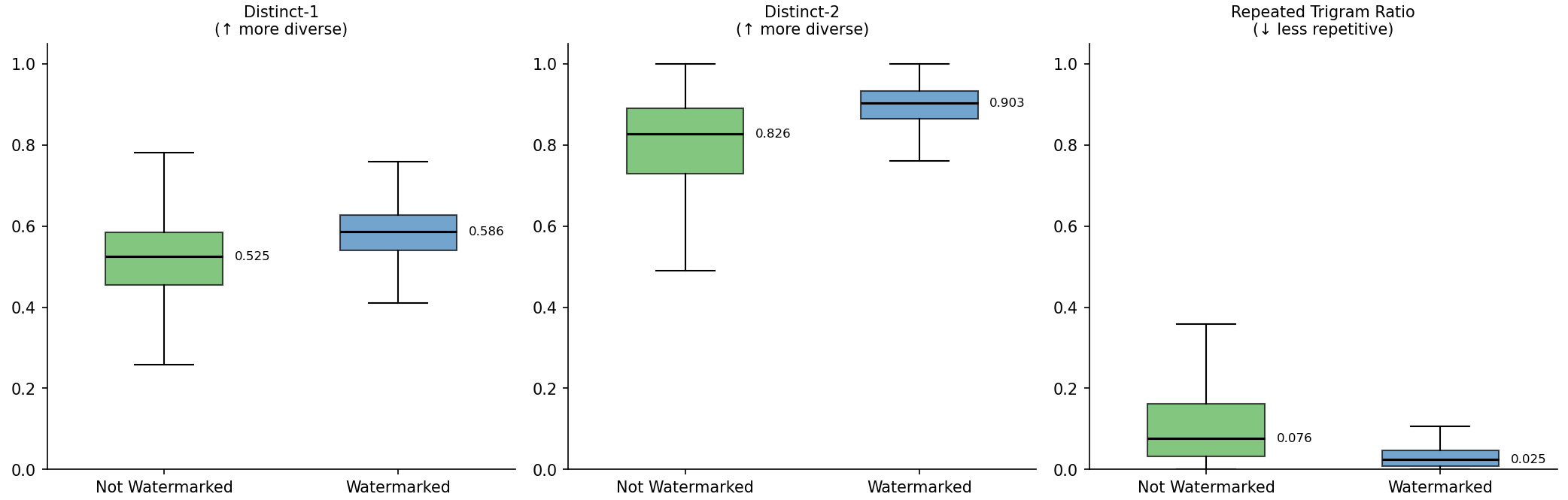}
  \caption{Lexical diversity distributions for 200-token Llama~2~7B generations on the OpenGen dataset~\cite{opengendataset}.
}
  \label{fig:lexical_diveristy}
\end{figure*}

\section{On Repetition Loops in Greedy Watermarked Generation}
\label{app:repetition_loops}

Repetition loops---degenerate outputs in which the model cyclically repeats a token or short text segment---are a well-documented failure mode of neural text generation that has been studied extensively in the literature independent of watermarking~\cite{li2023repetition,hiraoka-inui-2025-repetition}. 

In the logit-level watermarking setting most relevant to us, Fernandez et al.~\cite{Fernandez2023ThreeBricks} report that the green-list bias of Kirchenbauer et al.~\cite{kirchenbauer2023watermark} can push generation toward repeated tokens when the hash window is small, degrading text quality to the point of producing unusable completions. Mitigating neural text degeneration is thus an active and distinct research question that falls outside the scope of the present paper, whose focus is watermark embedding and recovery.

We did, however, examine whether simple countermeasures suggested by the watermarking literature help in our setting. Kirchenbauer et al.~\cite{kirchenbauer2023watermark} seed the green/red split from a window of $h$ preceding tokens, and Fernandez et al.~\cite{Fernandez2023ThreeBricks} observe that enlarging this window reduces the bias toward repeated tokens. We tested the analogous adjustment in our scheme but observed a limited effect on occurrence of loops.

Consistent with the literature, we do observe occasional repetition loops in our 200-token watermarked generations. A representative example, in which a coherent prefix degenerates into a looping segment, is:
\begin{quote}
\textit{The idea was scrapped because ``the dialogue was too long.'' This was one of the first SpongeBob episodes that used the ``Silly Sound Effects'' that are used to censor profane speech. This episode won the Animation Magazine's award for ``Best TV Production'', and was nominated by the Academy of Television Arts and Sciences for ``Outstanding Individual Achievement in Animation''. It was nominated by the Academy of Television Arts and Sciences for ``Outstanding Individual Achievement in Animation''. It was nominated by the Academy of Television Arts and Sciences for ``Outstanding Individual Achievement in Animation''. It was nominated by the Academy of Television Arts and Sciences for ``Outstanding Individual Achievement in Animation''.\ldots}
\end{quote}
Such loops are a known limitation of greedy decoding and of logit-biased generation more broadly. We conjecture that this behaviour is less prominent with larger, more powerful models, but we were unable to evaluate this hypothesis with the model scales available to us; confirming it, quantifying loop frequency, and developing targeted mitigations are left for future work.

\section{False Positive Rate (FPR) from Theorem~1}
\label{app:fpr}

Let

\[
    P=\{(x_i,y_i)\}_{i=1}^{N}\subset \GF(2^n)\times \GF(2^n)
\]

be the set of extracted points used by the detector (one point per decoded
block). The detector declares a watermark present if

\[
    \MCP(P)\ge \tau,
\]
where $\MCP(P)$ is the maximum number of points in $P$ lying on a single
(non-vertical) line and $\tau$ is the detection threshold.

\paragraph{Number of blocks and $n$}

If the verification segment contains $T$ tokens and each watermark block encodes $n$ bits, then the number of decoded blocks (hence extracted points) is typically

\begin{equation}
\label{eq:num-blocks}
    N \approx \left\lfloor \frac{T}{n}\right\rfloor .
\end{equation}

\paragraph{Definition of FPR}

Under the null hypothesis (no watermark), we model $P$ as $N$ i.i.d.\ uniform points in $\GF(2^n)^2$, as in Theorem~\ref{thm:rand_collinearity}. The false positive rate is

\begin{equation}
\label{eq:fpr-def}
\begin{aligned}
    \FPR(\tau;N,n)
    &= \Pr\!\left[\MCP(P)\ge \tau \,\middle|\, \text{no watermark}\right].
\end{aligned}
\end{equation}

\paragraph{FPR bound via Theorem~1}

Let $q=2^n$. Under this model, by Theorem~1, for $\tau\ge 3$,

\begin{align}
    \FPR(\tau;N,n)
    &\le \binom{N}{\tau}\, q^{-(\tau-2)} \notag\\
    &= \binom{N}{\tau}\, 2^{-n(\tau-2)} .
    \label{eq:fpr-bound}
\end{align}

\paragraph{Choosing $\tau$ for a target FPR}

Given a desired upper bound $\alpha$ (e.g., $\alpha=0.01$), choose the smallest
integer $\tau\ge 3$ such that

\begin{equation}
\label{eq:tau-choice}
\begin{aligned}
    \binom{N}{\tau}\, 2^{-n(\tau-2)} \le \alpha .
\end{aligned}
\end{equation}

Equivalently (taking $\log_2$),

\begin{equation}
\label{eq:tau-choice-log}
\begin{aligned}
    \log_2 \binom{N}{\tau} - n(\tau-2) \le \log_2 \alpha .
\end{aligned}
\end{equation}

\subsection{Worked example}

Assume $n=8$ (so $q=256$), $N=25$ blocks, and threshold $\tau=5$. Then

\begin{align}
    \FPR
    &\le \binom{25}{5}\,256^{-3}
    = 53130 \cdot 256^{-3} \notag\\
    &= 53130 \cdot \frac{1}{16777216}
    \approx 3.17\times 10^{-3}.
    \label{eq:ex1}
\end{align}

Thus, the bound predicts an FPR below $0.4\%$ for these parameters.

The bound in \eqref{eq:fpr-bound} is guaranteed only under the i.i.d.\ uniform model; on real text, the empirical FPR can exceed it (Appendix~\ref{app:fpr-empirical}).

\subsection{Selecting $\tau$ for the experimental settings}
\label{app:tau-table}
Table~\ref{tab:tau-selection} reports, for each watermark size used in our experiments ($W=2n$ bits), the smallest threshold $\tau\ge 3$ satisfying Eq.~\eqref{eq:tau-choice} for three target false-positive rates $\alpha$, with $T=208$ watermarked tokens and $N=\lfloor T/n\rfloor$ decoded blocks. These values are obtained directly from the analytical bound in Eq.~\eqref{eq:fpr-bound}. As $n$ (and hence the watermark size) grows, the required $\tau$ drops sharply, since accidental collinearities become negligible in larger fields. We recall that the match rates reported in Tables~\ref{tab:comparison} and~\ref{tab:llmw_results} are computed by exact coefficient recovery and do not invoke $\tau$; this table characterizes the presence-detection variant. The bound is conservative (union bound) relative to the i.i.d.\ uniform model of Theorem~\ref{thm:rand_collinearity}; on real text the empirical FPR can exceed it (Appendix~\ref{app:fpr-empirical}). The thresholds apply to the base detector without bit-correction variants.

\begin{table}[htbp]
\centering
\caption{Detection threshold $\tau$ (smallest $\tau\ge 3$ satisfying Eq.~\eqref{eq:tau-choice}) for each watermark size $W=2n$, with $T=208$ tokens and $N=\lfloor T/n\rfloor$ blocks, at three target FPR levels $\alpha$. Derived analytically from the union bound in Eq.~\eqref{eq:fpr-bound} under the i.i.d.\ uniform model; measured rates on real text are in Table~\ref{tab:fpr-empirical}. Applies to the base detector without bit-correction.}
\label{tab:tau-selection}
\resizebox{\columnwidth}{!}{
\begin{tabular}{|c|c|c|c|c|c|}
\hline
$W$ (bits) & $n$ & $N$ & $\alpha=10^{-2}$ & $\alpha=10^{-3}$ & $\alpha=10^{-4}$ \\
\hline
16 & 8  & 26 & 5 & 6 & 6 \\
20 & 10 & 20 & 4 & 5 & 5 \\
24 & 12 & 17 & 4 & 4 & 5 \\
32 & 16 & 13 & 3 & 4 & 4 \\
\hline
\end{tabular}
}
\end{table}

\subsection{Empirical Validation of the FPR Bound}
\label{app:fpr-empirical}

The bound in Eq.~\eqref{eq:fpr-bound} is a conservative union-bound estimate under the i.i.d.\ uniform model of Theorem~\ref{thm:rand_collinearity}. To test how well it carries over to real text, we measured the empirical false-positive rate on human-written text. We used 6,791 answers from the human-written subset of the HC3 dataset~\cite{guo-etal-2023-hc3}, each containing at least 208 tokens after tokenization with the Llama~2 tokenizer. Each text was decoded with a random secret key (one key per field size $n$, shared across all texts), so no watermark had been embedded. Because the texts are human-written, the null hypothesis holds for any key; however, the recovered points are only \emph{approximately} i.i.d.\ uniform in $\GF(2^n)^2$; the green/red classification of real, frequency-skewed tokens departs from the idealized uniform model of Theorem~\ref{thm:rand_collinearity}, and reusing a single key does not average out this departure. We then ran the MCP detector and recorded the largest collinear set.

Table~\ref{tab:fpr-empirical} reports the empirical FPR alongside the theoretical upper bound from Eq.~\eqref{eq:fpr-bound} for the detection thresholds used in Table~\ref{tab:tau-selection}. Two observations stand out. First, at the loosest operational thresholds (those selected for a target $\alpha=10^{-2}$), the empirical FPR is of the same order as the bound and, in all four cases, stays below the $10^{-2}$ target---so the bound is adequate for threshold selection at this operating point. Second, the empirical FPR \emph{exceeds} the theoretical upper bound in seven of the eight rows, so the bound is \emph{not} a strict upper bound on real text; it underestimates the true FPR, increasingly so as $\tau$ grows. Consequently, at the more stringent target $\alpha=10^{-4}$ the thresholds suggested by Eq.~\eqref{eq:fpr-bound} yield empirical FPRs above $10^{-4}$ for three of the four watermark sizes ($n=8,10,12$); only $n=16$ meets the target. We therefore recommend Eq.~\eqref{eq:fpr-bound} for threshold selection at loose targets ($\alpha\in\{10^{-2},10^{-3}\}$) and caution against relying on it at $\alpha=10^{-4}$, where the union-bound estimate underestimates the real-text FPR. This is consistent with real text violating the independence/uniformity assumption of Theorem~\ref{thm:rand_collinearity} discussed in the Limitations section. Finally, with $6{,}791$ texts the smallest non-zero rate we can resolve is $1/6791\approx1.47\times10^{-4}$; rows reported as $0$ (e.g., $n=16,\tau=4$) are thus consistent with any true FPR below this floor, and the single-count rows carry wide uncertainty, so the experiment is informative at the $\alpha=10^{-2}$ and $10^{-3}$ levels but underpowered to validate the $10^{-5}$--$10^{-8}$ predictions of the bound at large $\tau$.

\begin{table}[htbp]
\centering
\caption{Empirical vs. theoretical FPR on 6,791 human-written HC3 answers with $T=208$ tokens. The empirical FPR is the fraction of unwatermarked texts whose maximum collinear count reaches the threshold $\tau$. Rates below $1/6791\approx1.47\times10^{-4}$ are not resolvable with this sample size.}
\label{tab:fpr-empirical}
\resizebox{\columnwidth}{!}{
\begin{tabular}{|c|c|c|c|c|c|}
\hline
$W$ (bits) & $n$ & $N$ & $\tau$ & Theoretical FPR & Empirical FPR \\
\hline
16 & 8  & 26 & 5 & $3.92\times10^{-3}$ & $5.15\times10^{-3}$ \\
16 & 8  & 26 & 6 & $5.36\times10^{-5}$ & $2.95\times10^{-4}$ \\
20 & 10 & 20 & 4 & $4.62\times10^{-3}$ & $4.71\times10^{-3}$ \\
20 & 10 & 20 & 5 & $1.44\times10^{-5}$ & $1.47\times10^{-4}$ \\
24 & 12 & 17 & 4 & $1.42\times10^{-4}$ & $4.42\times10^{-4}$ \\
24 & 12 & 17 & 5 & $9.00\times10^{-8}$ & $1.47\times10^{-4}$ \\
32 & 16 & 13 & 3 & $4.36\times10^{-3}$ & $4.56\times10^{-3}$ \\
32 & 16 & 13 & 4 & $1.66\times10^{-7}$ & $0$ \\
\hline
\end{tabular}
}
\end{table}

\section{Quality of Watermarked Texts}

Our watermarking scheme, like other existing approaches, modifies logit values by adding a small bias to a subset of logits (green tokens). We use the same logit bias as the state-of-the-art method proposed by \cite{Qu:2024aa}, and their analysis can be directly applied to our approach. Specifically, they show that using a watermark with bias $\delta = 6$ only slightly degrades text quality, as measured by perplexity, and that the perplexity of LLM-generated text is similar with and without watermarking. In addition, they conducted a small qualitative study with human judges, concluding that any text quality degradation caused by the watermark is subtle according to human evaluation.

\section{Computational complexity of MCP-based detection and practical optimizations}
\label{app:mcp-complexity}

\subsection{Baseline complexity}

A standard approach to MCP in the plane over a finite field is the
``anchor-and-count-slopes'' method: for each anchor point $p_i$, compute slopes
to all other points and count duplicates (points sharing the same slope with
$p_i$ lie on the same line through $p_i$). Using hashing for slope counts, this
runs in $\mathcal{O}(N^2)$ time and $\mathcal{O}(N)$ extra memory, which is already efficient for moderate $N$ (e.g., tens to a few hundreds of
points), but could become expensive if $N$ is allowed to scale with very long
texts.

\subsection{Optimization 1: fixed verification window (constant-time in text length)}

A simple and effective mitigation is to cap the number of blocks inspected by
running detection only on a fixed-length token window of size $T_0$
(e.g., $T_0=200$ tokens). This bounds the number of points to

\[
    N_0 \approx \left\lfloor \frac{T_0}{n}\right\rfloor ,
\]
and thus bounds MCP computation to $\mathcal{O}(N_0^2)$, which is constant with
respect to total document length. The remainder of a long text can be treated as
redundancy: it does not increase worst-case inference cost, while still being
useful for robustness (see below).

\subsection{Optimization 2: multi-window/redundancy for robustness}

Using only a single fixed window may reduce detection robustness if that segment
is heavily edited. A robust variant keeps inference efficient by evaluating MCP
on $K$ non-overlapping (or randomly sampled) windows of length $T_0$ and taking
the maximum score:

\[
    S = \max_{j=1,\dots,K} \MCP(P_j),
\]

where $P_j$ are points extracted from window $j$. This costs
$\mathcal{O}(K N_0^2)$, still constant in total text length for fixed $K$ and
$T_0$. In practice, small $K$ (e.g., $K\in\{2,3,5\}$) provides strong robustness
while remaining fast. Taking the maximum over $K$ windows runs $K$ tests, however, so the false-positive rate can grow up to $K$-fold; $\tau$ should then be chosen from Eq.~\eqref{eq:tau-choice} with $\alpha/K$ in place of $\alpha$.

\subsection{Optimization 3: early stopping with threshold $\tau$}

Detection only requires deciding whether $\MCP(P)\ge \tau$. The slope-counting
procedure can stop early as soon as any anchor point yields a line with at least
$\tau$ points. Conversely, if an anchor point has fewer than $\tau-1$ remaining
points, it cannot create a $\tau$-point line and can be skipped. These checks
often reduce runtime substantially in both positive and negative cases.

\section*{Acknowledgment}
The authors disclose the use of AI assistants from the GPT (OpenAI) and Claude (Anthropic) model families for software implementation and editing prose. All content was reviewed and verified by them. Code is available at \url{https://github.com/jaroslawjanas/llmw-lagrange}

\bibliographystyle{IEEEtran}
\bibliography{LLM,custom}

\end{document}